\documentclass[transmag]{IEEEtran}
\usepackage{latexsym}
\usepackage{graphicx}
\usepackage{amsfonts,amssymb,amsmath}

\usepackage{amsmath,graphicx}

\usepackage[11pt]{moresize}

\usepackage{dsfont}
\usepackage{xcolor}

\usepackage[colorlinks=true,linkcolor=blue,citecolor=red,hyperfootnotes,pdftex]{hyperref}

\usepackage{blindtext}

\usepackage{latexsym}
\usepackage{amssymb}

\usepackage{amsfonts}
\usepackage[mathscr]{eucal}
\usepackage{amsthm}

\usepackage{array}
\usepackage{multirow}

\usepackage{float}
\usepackage{url}
\usepackage{eurosym}
\usepackage{setspace}
\usepackage{balance}
\usepackage{dsfont}
\usepackage{bbm}
\usepackage{array}
\usepackage{caption}
\renewcommand\thetable{\arabic{table}}

\usepackage{multirow}

\DeclareFontFamily{U}{mathx}{\hyphenchar\font45}
\DeclareFontShape{U}{mathx}{m}{n}{<5> <6> <7> <8> <9> <10> <10.95> <12> <14.4> <17.28> <20.74> <24.88> mathx10}{}
\DeclareSymbolFont{mathx}{U}{mathx}{m}{n}
\DeclareMathSymbol{\bigtimes}{1}{mathx}{"91}

\DeclareMathAlphabet{\mts}{U}{rsfs}{m}{n}

\theoremstyle{definition}
\newtheorem{Def}{Definition}

\theoremstyle{plain}
\newtheorem{Thm}{Theorem}
\newtheorem{Cor}{Corollary}

\newtheorem{Lem}{Lemma}
\newtheorem{Asm}{Assumption}

\theoremstyle{remark}
\newtheorem{Rem}{Remark}
\newtheorem{Exa}{Example}

\usepackage{xcolor}

\newcommand{\qedsymb}{\hfill\ensuremath{\blacksquare}}                 % end of proofs

\def\BibTeX{{\rm B\kern-.05em{\sc i\kern-.025em b}\kern-.08em T\kern-.1667em\lower.7ex\hbox{E}\kern-.125emX}}
\allowdisplaybreaks

\begin{document}

\title{Deception in Social Learning}

\author{Konstantinos Ntemos, Virginia Bordignon, \IEEEmembership{Student Member, IEEE}, \\ 
Stefan Vlaski, \IEEEmembership{Member, IEEE},
and Ali H. Sayed, \IEEEmembership{Fellow, IEEE}\thanks{
                This work was supported in part by the Swiss National Science Foundation grant 205121-184999. All authors are with the School of Engineering, Ecole Polytechnique F\'ed\'erale de Lausanne (EPFL). E-mails: konstantinos.ntemos@epfl.ch, virginia.bordignon@epfl.ch, stefan.vlaski@epfl.ch,  ali.sayed@epfl.ch. A limited short version of this work appears in \cite{Ntemos_2021}.}
}

\maketitle

\begin{abstract}
A common assumption in the social learning literature is that agents exchange information in an unselfish manner. In this work, we consider the scenario where a subset of agents aims at \textit{deceiving} the network, meaning they aim at driving the network beliefs to the wrong hypothesis. The adversaries are unaware of the true hypothesis. However, they will ``blend in" by behaving similarly to the other agents and will manipulate the likelihood functions used in the belief update process to launch {\em inferential attacks}. We will characterize the conditions under which the network is misled. Then, we will explain that it is possible for such attacks to succeed by showing that strategies exist that can be adopted by the malicious agents for this purpose. We examine both situations in which the agents have access to information about the network model as well as the case in which they do not. For the first case, we show that there always exists a way to construct fake likelihood functions such that the network is deceived regardless of the true hypothesis. For the latter case, we formulate an optimization problem and investigate the performance of the derived attack strategy by establishing conditions under which the network is deceived. We illustrate the learning performance of the network in the aforementioned adversarial setting via simulations. In a nutshell, we clarify when and how a network is deceived in the context of non-Bayesian social learning.
\end{abstract}

% Note that keywords are not normally used for peerreview papers.

\section{INTRODUCTION}

\IEEEPARstart{T}{he} {\em social learning} paradigm refers to the setting where agents aim at learning an underlying state, by means of their local observations as well as information provided by their neighbors. 
The communication protocol among the agents is driven by an underlying graph topology where each agent is only allowed to communicate directly with its neighbors. A special feature of social learning is that the agents do not have access to the raw data of their neighbors, but only to processed variables provided by these neighbors. Thus, social learning entails an implicit {\em inference} problem in that agents need to 
have to reason about how the received information was generated by their peers. This situation emerges in many areas of interest, ranging from technical applications to social sciences.

The notion of {\em rationality} implies that agents perform {\em Bayesian inference} to reason about the information they receive \cite{Savage_1954}. 
However, performing Bayesian inference over social networks is challenging because knowledge about each agent's observation model and the entire network topology is required. 
A series of works that relax this requirement and deviate from the fully rational model have been pursued in the literature \cite{De_Groot_1974, Jadbabaie_2012, Zhao_2012, Molavi_2018,Lalitha_2014,Lalitha_2018,Nedic_2017,Salami_2017,Virginia_2020} to provide tractable alternatives to fusing information among agents in scenarios of {\em incomplete information}.

One of the first such models is De Groot's model \cite{De_Groot_1974}, followed by other more sophisticated models \cite{Jadbabaie_2012,Zhao_2012,Molavi_2018,Lalitha_2014,Salami_2017,Lalitha_2018,Nedic_2017,Virginia_2020}. A common characteristic of the proposals in \cite{De_Groot_1974,Jadbabaie_2012,Zhao_2012,Salami_2017} is the fact that each agent fuses their own belief with the beliefs provided by their neighbors using a {\em linear combination rule} (weighted arithmetic mean). The authors in \cite{Molavi_2018} follow an axiomatic approach to formalize the theoretical grounds of non-Bayesian learning and highlight the deviations from rational learning. They impose some behavioral assumptions that lead to a social learning rule that has a {\em log-linear} fusion form. The log-linear rule \cite{Lalitha_2014} leads to faster {\em learning rate} of the unknown hypothesis \cite{Lalitha_2018}, \cite{Nedic_2017}.

The standard assumption in most of these works on social learning is that agents %in an {\em honest} fashion
exchange their beliefs in a cooperative fashion. However, in many cases, agents may exhibit intentional {\em misbehavior} or operate in a {\em faulty} way. For this reason, social learning in {\em untrusted} environments, where malicious entities may be present, has been receiving increasing attention  \cite{Su_2018,Vyavahare_2019,Bhotto_2018,Hare_2019}. Our work is a contribution to this line of research by considering the scenario where a collection of {\em malicious} agents or {\em adversaries} aim at forcing the network to accept a wrong hypothesis by causing the network beliefs to converge to the wrong state. What is particular about our approach is that the adversaries are not assumed to have knowledge of the true hypothesis. Moreover, they follow similar protocols to the other agents in the network in an effort to ``blend in'' and evade detection.

In order to clarify the limits of performance we examine two scenarios. First, we assume the adversaries have knowledge of the network structure and other agents' observation models. In this case, we show that it is always possible for them to mislead the other agents to the wrong belief by ``blending in''. Second, we assume the adversaries do not have information about the network structure and derive an attack strategy by formulating and solving an optimization problem meant to increase the likelihood of false beliefs by the other agents. We carry out theoretical investigations of the behavior of the network under both scenarios and illustrate the results by means of computer simulations.
\subsection{Prior works}
Some of the earlier studies of robustness of distributed processing systems when some components suffer malfunction or exhibit adversarial behavior,  are the works by  \cite{Lamport_1982}, \cite{Strong_1982} on %and the study of robustness against 
{\em Byzantine attacks}. Under these types of attacks, the misbehaving components are allowed to deviate from the protocol followed by all other agents in some {\em arbitrary} way and no constraints are imposed on their (mis)behavior.

There have also been works on  {\em robust detection} \cite{Huber_1965}, where adversaries aim at driving a fusion center to the wrong decision with extensions appearing in \cite{Marano_2008,Kailkhura_2013}. Reference \cite{Marano_2008} studies distributed detection with malicious agents. In this context, a fusion center receives observations from the dispersed agents to perform a binary hypothesis test. A subset of the agents are malicious and select their likelihood functions to maximize the probability of erroneous detection. The work investigates attack strategies by the adversaries. This work was extended in \cite{Kailkhura_2013}, which considered the scenario where there exists a mechanism to discriminate between normal and malicious agents. A useful discussion on robust detection and estimation can be found in \cite{Vempaty_2013}. The aforementioned works focus on the so-called {\em parallel setup} where the agents send their information to a fusion center and try to impede its detection performance. In contrast to the parallel setup, in this paper we focus on the problem of learning in the presence of malicious/faulty agents within the framework of non-Bayesian social learning over networks. This problem poses new challenges, as each agent tries to learn the state in a decentralized fashion (each agent communicates only with its neighbors and information diffuses over a graph). 

Social learning in untrusted environments, where there is a subset of agents acting in a malicious/faulty manner to mislead normal agents, is 
studied in \cite{Su_2018,Vyavahare_2019,Bhotto_2018,Hare_2019}. The work \cite{Su_2018} designs algorithms that enable the normal agents to learn the true state, despite the misinformation provided by Byzantine agents. The authors study the learning consistency of their algorithms and show that under certain assumptions on the structure of the communication graph, the normal agents successfully identify the underlying state. This work was extended to prove consistent learning under more relaxed assumptions on agents' connectivity in \cite{Vyavahare_2019}. However, the required conditions on the network topology continue to be strict. In particular, each agent is required to trim away a subset of receiving beliefs in proportion to an upper bound on the number of malicious agents in the network. This implies that the size of the local neighborhoods has to be larger than the number of malicious agents to avoid isolating any agents. Moreover, Byzantine adversaries are usually assumed to have knowledge of the system attributes, including the network model and the true hypothesis. In this work, we investigate scenarios where adversaries do not have access to such knowledge.

Apart from Byzantine attacks, which have great freedom in behaving in arbitrary manner,  other works in the literature have focused on other types of attacks. 
One such scenario, which is the subject of interest in this paper, is to simply focus on what is ``minimally'' needed to drive a network to a wrong belief or to a particular wrong belief. This cannot be guaranteed by letting adversaries send arbitrary information, since it can result in undesired non-convergent behavior. One strategy for adversarial behavior to mislead the network in a more controlled manner is to use {\em corrupted} likelihood functions. The main challenge is to design these functions. This type of adversarial behavior was considered in \cite{Bhotto_2018,Hare_2019}. For example, the work \cite{Hare_2019} discusses a variety of attack scenarios and discriminates between {\em weak} and {\em strong} malicious agents. Weak malicious agents manipulate only the likelihood functions in the belief update rule, while strong malicious agents can additionally filter out information sent from honest agents.

In this paper we consider the weak malicious agents case. We assume they follow the same protocol as the other agents so that they appear to ``blend in". This scenario also captures the case where the likelihood functions are provided to the agents exogenously and a malicious actuator intentionally modifies some agents' likelihood functions to launch a {\em stealth attack} (in the sense that a network agent may operate according to the proposed protocol but its observation model might have been corrupted/manipulated). Although this kind of attacks was discussed in \cite{Bhotto_2018,Hare_2019}, where adversaries use corrupted likelihood functions to update their beliefs, the design of these adversarial likelihood functions or their impact on the learning performance were not investigated. The focus of these works is on detection of  adversaries. Instead, our work is focused on devising adversarial strategies that force the network to be deceived and on investigating their impact on the learning performance of the network.
\subsection{Our Contribution}
This paper addresses the case where the adversaries are {\em agnostic}, meaning they do not have knowledge about the true state and aim at forcing the network beliefs to the wrong hypothesis. We assume that adversaries participate in the information diffusion process as dictated by the social learning protocol, but disseminate falsified beliefs, which are produced by the use of corrupted likelihood functions. We refer to this type of attacks as ``inferential attacks", due to the fact that adversaries have no knowledge of the true state and try to drive the network beliefs to the wrong state by manipulating their inference model (i.e, likelihood functions).

More specifically, in this work, we answer the following questions. First, under what conditions can an unguarded network (no detection mechanism is employed) be misled under inferential attacks? Second, if adversaries do not know the true state, is there a way to construct fake likelihood functions that drive the normal agents' beliefs to the wrong state, given some knowledge of the network properties? Finally, in scenarios of incomplete information, when adversaries do not have any knowledge about the network, how should they manipulate the observation models to mislead the network?

We characterize the conditions under which the network is misled. We prove that it depends on the agents' observation models, malicious and benign agents' centrality, and attack strategies. In this way, we reveal an interplay between the network topology, which captures the diffusion of information, and injection of mis-information in the social learning paradigm. Then, we prove that if an adversary knows certain network characteristics, then there is always a way to construct an attack strategy that misleads the network for every possible true hypothesis. Finally, we study the scenario when adversaries have no knowledge about the network properties.  We propose an attack strategy and investigate its impact on the learning performance. Summarizing our contribution on a high level, we answer the questions of when and how a network is deceived in the context of non-Bayesian social learning.
\subsection{Notation}
We use boldface letters to denote random variables and normal letters to denote their realizations. For a random variable $\boldsymbol{x}$, we denote the KL divergence from distribution $L_1(\boldsymbol{x})$ to distribution $L_2(\boldsymbol{x})$ by $D_{KL}(L_1||L_2)$. $\mathds{1}$ denotes the column vector whose every element is equal to $1$ and $[A]_{\ell k}$ corresponds to the element at row $\ell$ and column $k$ of the matrix $A$.

\section{System Model}
We assume a set $\mathcal{N}=\mathcal{N}^n\bigcup\mathcal{N}^m$ of agents, where $\mathcal{N}^n$ and $\mathcal{N}^m$ denote the sets of normal and malicious agents, respectively. The types of the agents (i.e., normal or malicious) are unknown.  Without loss of generality, we index first the malicious agents, followed by the normal ones. The network is represented by an undirected graph $\mathcal{G}=\langle \mathcal{N}, \mathcal{E} \rangle$, where $\mathcal{E}$ includes bidirectional links between agents. The set of neighbors of an agent $k$ is denoted by  $\mathcal{N}_k$. 

We consider an adversarial setting where the normal agents aim at learning the true state $\theta^{\star}\in\Theta=\{\theta_1,\theta_2\}$, while malicious agents try to impede the normal agents by forcing their beliefs towards the wrong state. All agents are unaware of the true state $\theta^{\star}$.

We assume that each agent $k$ has access to observations $\boldsymbol{\zeta}_{k,i}\in\mathcal{Z}_k$ at every time $i\geq1$. Agent $k$ also has access to the likelihood functions $L_k(\zeta_{k,i}|\theta)$, $\theta\in\Theta$. The signals $\boldsymbol{\zeta}_{k,i}$ are independent and identically distributed (i.i.d.) over time. The sets $\mathcal{Z}_k$ are assumed to be finite with $|\mathcal{Z}_k|\geq2$ for all $k\in\mathcal{N}$. We will use the notation $L_k(\theta)$ instead of $L_k(\boldsymbol{\zeta}_{k,i}\vert\theta)$ whenever it is clear from the context. 
\begin{Asm}{\bf (Finiteness of KL divergences)}. 
\label{obs_independence}
For any agent $k\in\mathcal{N}$ and for any $\theta\neq\theta^{\star}$,  $D_{KL}\Big(L_k(
\theta^{\star})||L_k(
\theta)\Big)$ 
is finite.
\qedsymb\end{Asm}
At each time $i$, agent $k$ keeps a {\em belief vector} $\boldsymbol{\mu}_{k,i}$, which is a probability distribution over the possible states. The belief component $\boldsymbol{\mu}_{k,i}(\theta)$ quantifies the confidence of agent $k$ that $\theta$ is the true state. Since all agents, both normal and malicious, are unaware of the true state, we impose the following assumption on initial beliefs.
\begin{Asm} {\bf (Positive initial beliefs)}. 
\label{positive_beliefs}
    $\mu_{k,0}(\theta)>0, \forall \theta\in\Theta,k\in\mathcal{N}$.
\qedsymb
\end{Asm}
\section{Social Learning with Adversaries}
Each agent $k$ uses the acquired observations $\boldsymbol{\zeta}_{k,i}$, along with the likelihood function $L_k(\zeta_{k, i}|\theta)$, to update their belief vector using Bayes' rule. Agents communicate with each other and exchange information. We consider the log-linear social learning rule \cite{Lalitha_2014,Virginia_2020,Matta_2020} where the normal agents update their beliefs in the following manner:
\begin{align}
\label{adapt}
    &\boldsymbol{\psi}_{k,i}(\theta)=
    \frac{L_k(\boldsymbol{\zeta}_{k,i}|\theta)\boldsymbol{\mu}_{k,i-1}(\theta)}{\sum_{\theta'}L_k(\boldsymbol{\zeta}_{k,i}|\theta')\boldsymbol{\mu}_{k,i-1}(\theta')},\quad k\in\mathcal{N}^n\\
 \label{combine}
    &\boldsymbol{\mu}_{k,i}(\theta)=\frac{\prod_{\ell\in\mathcal{N}_k}\boldsymbol{\psi}^{a_{\ell k}}_{\ell,i}(\theta)}{\sum_{\theta'}\prod_{\ell\in\mathcal{N}_k}\boldsymbol{\psi}^{a_{\ell k}}_{\ell,i}(\theta')}, \quad k\in\mathcal{N}^n
\end{align}
where $a_{\ell k}$ denotes the {\em combination weight} assigned by agent $k$ to neighboring agent $\ell$, satisfying $0<a_{\ell k}\leq1$, for all $\ell\in\mathcal{N}_k$, $a_{\ell k}=0$ for all $\ell\notin\mathcal{N}_k$ and $\sum_{\ell\in\mathcal{N}_k}a_{\ell k}=1$. Let $A$ denote the {\em combination matrix} which consists of all agents' combination weights with $[A]_{\ell k}=a_{\ell k}$. Clearly, $A$ is left-stochastic. Regarding the network topology, we impose the following assumption.
\begin{Asm}{\bf{(Strongly-connected network)}.}
\label{strongly_connected}
The communication graph is {\em strongly connected} (i.e., there always exists a path with positive weights linking any two agents and at least one agent has a self-loop, meaning that there is at least one agent $k\in\mathcal{N}$ with $a_{kk}>0$).\qedsymb
\end{Asm}
For a strongly connected network, the limiting behavior of $A^{\mathsf{T}}$ is given by $\lim_{i\rightarrow\infty}(A^{\mathsf{T}})^i=\mathds{1}u^{\mathsf{T}}$, where $u$ is the Perron eigenvector \cite{Sayed_2014} associated with the eigenvalue at $1$ and all its entries are positive and are normalized to add up to one. Moreover, its $k-$th entry $u_k$ expresses a measure of influence of agent $k$ on the network and is also called the  {\em centrality} of agent $k$.

We consider the scenario where adversaries aim at misleading the network to accept the wrong hypothesis by modifying the way they use their observations. More specifically, we assume that malicious agents deviate in step \eqref{adapt} by using a fake likelihood function, denoted by $\widehat{L}_k(\cdot)$ instead of $L_k(\cdot)$ to update their beliefs, while they follow \eqref{combine} 
without deviation. Inferential attacks are therefore modeled by assuming that adversaries follow the following update rule:
\begin{align}
         \label{adapt_malicious}
    &\boldsymbol{\psi}_{k,i}(\theta)=
    \frac{\widehat{L}_k(\boldsymbol{\zeta}_{k,i}|\theta)\boldsymbol{\mu}_{k,i-1}(\theta)}{\sum_{\theta'}\widehat{L}_k(\boldsymbol{\zeta}_{k,i}|\theta')\boldsymbol{\mu}_{k,i-1}(\theta')},\quad k\in\mathcal{N}^m.
\end{align}
The network model we consider and the interactions between a normal agent $k$ and an adversary $\ell$ which launches an inferential attack, are illustrated in Fig. \ref{net_img}. Agents $k,\ell$ exchange their intermediate beliefs $\boldsymbol{\psi}_{k,i},\boldsymbol{\psi}_{\ell,i}$ where we write $\boldsymbol{\psi}_{k,i}(L_k(\cdot)),\boldsymbol{\psi}_{\ell,i}(\widehat{L}_{\ell,i}(\cdot))$ to explicitly state that the adversary's shared beliefs $\boldsymbol{\psi}_{\ell,i}$ depend on the fake likelihood functions $\widehat{L}_{\ell,i}(\cdot)$ instead of the true ones $L_{\ell,i}(\cdot)$.
\begin{figure}[!h]
\centering
\includegraphics[width=0.35\textwidth]{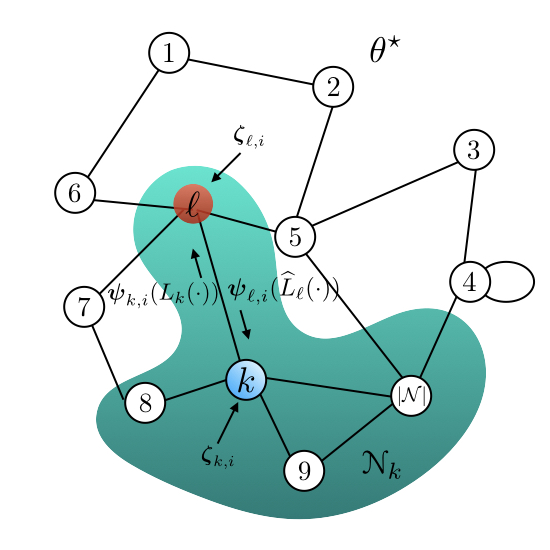}
\caption{Illustration of the network model and the interactions between a normal agent ($k$) and an adversary ($\ell$).}
\label{net_img}
\end{figure}

We impose the following technical assumption on the distorted likelihood functions.
\begin{Asm}
\label{finiteness}
{\bf(Distorted likelihood functions with full support)}. For every agent $k\in\mathcal{N}^m$, the distorted likelihood function satisfies $\epsilon\leq\widehat{L}_k(\zeta_{k,i}|\theta)$ for all $\zeta_{k,i}\in\mathcal{Z}_k$, $\theta\in\Theta$, where $0<\epsilon\ll1$ is a small positive real constant that satisfies
\begin{align}
\label{condition_e}
    \epsilon<\min_k\frac{1}{|\mathcal{Z}_k|}.
\end{align}\qedsymb
\end{Asm}
We say that an agent $k$'s belief converges \textit{almost surely} (a.s.) to the true state if $\boldsymbol{\mu}_{k,i}(\theta^{\star})\to1$ as $i\to\infty$ with probability $1$. Conversely, agent $k$'s belief converges a.s. to the wrong state if $\boldsymbol{\mu}_{k,i}(\theta^{\star})\to0$ as $i\to\infty$ with probability $1$. The following result characterizes the asymptotic learning behavior of the network.
\begin{Thm}{\bf (Belief convergence with adversaries)}.
\label{inconsistent_learning}
Under Assumptions \ref{obs_independence}, \ref{positive_beliefs}, \ref{strongly_connected}, \ref{finiteness}, two situations can arise%the following are true
:
\begin{enumerate}
    \item 
    All agents' beliefs converge a.s. to the wrong state if 
    \begin{align}
    {\begin{small}{
    \label{lem_ratiosm}
    \hspace{-8mm}\sum_{k\in\mathcal{N}^n}\hspace{-1.5mm}u_k\mathbb{E}\Bigg\{\log\frac{{L_k(\boldsymbol{\zeta}_k|\theta^{\star})}}{{L_k(\boldsymbol{\zeta}_k|\theta)}}\Bigg\}\hspace{-1mm}<\hspace{-1mm}
    \sum_{k\in\mathcal{N}^m}\hspace{-1.5mm}u_k\mathbb{E}\Bigg\{\log\frac{\widehat{L}_k(\boldsymbol{\zeta}_k|\theta)}{\widehat{L}_k(\boldsymbol{\zeta}_k|\theta^{\star})}\Bigg\}.
}\end{small}}
    \end{align}
   \item All agents' beliefs converge a.s. to the true state if
    \begin{align}
    {\begin{small}{
    \label{lem_ratiosm2}
    \hspace{-8mm}\sum_{k\in\mathcal{N}^n}\hspace{-1.5mm}u_k\mathbb{E}\Bigg\{\log\frac{{L_k(\boldsymbol{\zeta}_k|\theta^{\star})}}{{L_k(\boldsymbol{\zeta}_k|\theta)}}\Bigg\}\hspace{-1mm}>\hspace{-1mm}
    \sum_{k\in\mathcal{N}^m}\hspace{-1.5mm}u_k\mathbb{E}\Bigg\{\log\frac{\widehat{L}_k(\boldsymbol{\zeta}_k|\theta)}{\widehat{L}_k(\boldsymbol{\zeta}_k|\theta^{\star})}\Bigg\}
    }\end{small}}
    \end{align}
\end{enumerate}
where $\theta^{\star},\theta\in\Theta$, $\theta^{\star}\neq\theta$.
\end{Thm}
\begin{proof}
\textit{See} Appendix \ref{a1}.
\end{proof}
The Theorem characterizes under what condition the agents in the graph can be misled, namely, when condition \eqref{lem_ratiosm} holds. Thus, malicious agents would strive to construct their distorted likelihood functions to satisfy \eqref{lem_ratiosm}. The expectation in \eqref{lem_ratiosm} and \eqref{lem_ratiosm2} is taken with respect to (w.r.t.) the true likelihood distributions, $L_k(\boldsymbol{\zeta}_k|\theta^{\star})$. Since $\boldsymbol{\zeta}_{k,i}$ are i.i.d. over time, we omit the time index $i$. The threshold rule \eqref{lem_ratiosm}-\eqref{lem_ratiosm2} fully characterizes the convergence of network beliefs. Note that whether or not the agents' beliefs will converge to the true state depends on the agents' observation models (informativeness of the signals), on the distorted likelihood functions, and on the network topology (agents' centrality).

Relation \eqref{lem_ratiosm} can be expressed in terms of relative entropy measures as follows:
{\begin{small}{
\begin{align}
    \label{kl_inequality}
        &\sum_{k\in\mathcal{N}^n}u_kD_{KL}\Bigl(L_k(
        \theta^{\star}))||L_k(
        \theta)\Bigr)<\nonumber\\
        &\sum_{k\in\mathcal{N}^m}u_kD_{KL}\Bigl(L_k(
        \theta^{\star})||\widehat{L}_k(
        \theta^{\star})\Bigr)
        -
    \sum_{k\in\mathcal{N}^m}u_kD_{KL}\Bigl(L_k(
    \theta^{\star})||\widehat{L}_k(
    \theta)\Bigr).
\end{align}
}\end{small}}%
Condition \eqref{kl_inequality} suggests that from the malicious agents' perspective, for a given $\theta^{\star}\in\Theta$, the distorted likelihood function given the true state $\widehat{L}_k(\boldsymbol{\zeta}_k|\theta^{\star})$ should be quite {\em different} from the true likelihood function $L_k(\boldsymbol{\zeta}_k|\theta^{\star})$, while the distorted likelihood function for the false state $\widehat{L}_k(\boldsymbol{\zeta}_k|\theta)$ should be similar to the true likelihood function corresponding to the true state $L_k(\boldsymbol{\zeta}_k\vert\theta^{\star})$ (since KL divergence is nonnegative). 

Clearly, the family of distorted likelihood functions $\widehat{L}_k(\cdot\vert\theta_1),\widehat{L}_k(\cdot\vert\theta_2)$ that satisfy \eqref{lem_ratiosm}, or equivalently \eqref{kl_inequality}, for both cases when $\theta^{\star}=\theta_1$,  $\theta^{\star}=\theta_2$ will successfully deceive the network no matter what the true state $\theta^{\star}$ is. In the next Section, we investigate the construction of such PMFs that enable agnostic adversaries (i.e., they do not know what the true state $\theta^{\star}$ is) to successfully deceive the network.
\subsection{Attack strategies with known network divergences}
We now examine the question of whether and how adversaries can construct $\widehat{L}_k(\cdot\vert\theta_1),\widehat{L}_k(\cdot\vert\theta_2)$ in such a way that the network will always be driven to the wrong hypothesis no matter what the true state $\theta^{\star}$ is. Note that the true state is unknown to the adversaries as well. Thus, adversaries should select $\widehat{L}_k(\cdot\vert\theta_1),\widehat{L}_k(\cdot\vert\theta_2)$, $k\in\mathcal{N}^m$ such that \eqref{lem_ratiosm}, or equivalently \eqref{kl_inequality}, is satisfied for every possibility for $\theta^{\star}\in\Theta$ to ensure that the network will converge to the wrong hypothesis always (i.e., adversaries force the network beliefs to $\theta_1$ if $\theta^{\star}=\theta_2$ and to $\theta_2$ if $\theta^{\star}=\theta_1$).

Let us define the following quantities:
\begin{align}
    &S_j\triangleq\sum_{k\in\mathcal{N}^n}\hspace{-1.5mm}u_k\mathbb{E}\Bigg\{\log\frac{{L_k(\boldsymbol{\zeta}_k|\theta_j)}}{{L_k(\boldsymbol{\zeta}_k|\theta_{j'})}}\Bigg\}\\
    &R_{k,j}\triangleq u_k\sum_{\zeta_k\in\mathcal{Z}_k}L_k(\zeta_k|\theta_j)\log\frac{\widehat{L}_k(\zeta_k|\theta_{j'})}{\widehat{L}_k(\zeta_k|\theta_j)},\quad k\in\mathcal{N}^m
\end{align}
where $\theta_j=\theta^{\star}$, $j,j'\in\{1,2\},j\neq j'$. $S_j$ denotes the term on the left-hand side (LHS) of  \eqref{lem_ratiosm} for $\theta^{\star}=\theta_j, j=1,2$. 
We call $S_j$  \textit{normal sub-network divergence}, or simply \textit{divergence} of the normal sub-network for $\theta^{\star}=\theta_j$. $R_{k,j}$ corresponds to adversary $k$'s contribution to the right-hand side (RHS) of \eqref{lem_ratiosm} for the case $\theta^{\star}=\theta_j$. Then, we can rewrite \eqref{lem_ratiosm} as
\begin{align}
\label{inequality_re}
    S_j
    <
    \sum_{k\in\mathcal{N}^m}R_{k,j},\quad j=1,2.
\end{align}
Let us first examine the following system of inequalities for an adversary $k\in\mathcal{N}^m$ (which is sufficient condition for \eqref{inequality_re} to hold if it holds for every adversary $k\in\mathcal{N}^m$, since $S_j\geq0$): 
    \begin{align}
    \label{lem_ratiosm30}
    S_j
    <
    R_{k,j}
    ,\quad k\in\mathcal{N}^m,\,j=1,2.
\end{align}
\begin{Rem}
\label{suf_condition}
Note that $S_j$, $j=1,2$ is a positive weighted sum of KL divergences (due to Assumption \ref{strongly_connected}, $u$ has positive entries) and as a result $S_j\geq0$. Thus, if \eqref{lem_ratiosm30} holds for all $k\in\mathcal{N}^m$ for $j$ such that $\theta_j=\theta^{\star}$, then \eqref{lem_ratiosm} holds as well.
\end{Rem}
We note that a Probability Mass Function (PMF) is \emph{uninformative} if  the likelihood functions are identical for both states, meaning $L_k(\zeta_k\vert\theta_1)=L_k(\zeta_k\vert\theta_2)$ for all $\zeta_k\in\mathcal{Z}_k$, otherwise the PMF is {\em informative}. Evidently, an agent with uninformative PMFs cannot discriminate between the two states, meaning that it cannot learn the underlying true state $\theta^{\star}$ by using only its own observations. Next, we establish that adversaries with uninformative PMFs cannot mislead the network for both possibilities $\theta^{\star}=\theta_1$ and $\theta^{\star}=\theta_2$, thus highlighting the limitations on an adversary's deceptive capabilities by its actual observation model (true likelihood functions).
\begin{table*}[!h]
% increase table row spacing, adjust to taste
\renewcommand{\arraystretch}{1.3}
%if using array.sty, it might be a good idea to tweak the value of
% \extrarowheight as needed to properly center the text within the cells
\caption{An example of the construction of fake likelihood functions for an adversary $k$ with $|\mathcal{Z}_k|=4$.}
\label{construction_table}
\centering
% Some packages, such as MDW tools, offer better commands for making tables
% than the plain LaTeX2e tabular which is used here.
% \begin{small}
\begin{tabular}{@{}|c||c|c|c|c|@{}}
\hline
&$\zeta^1_k$&$\zeta^2_k$&$\zeta^3_k$&$\zeta^4_k$\\
\hline
$\theta^{\star}=\theta_1$ & $\widehat{L}_k(\zeta^1_k\vert\theta_1)=1-2\epsilon-p_{k,2}$ & $\widehat{L}_k(\zeta^2_k\vert\theta_1)=p_{k,2}$ & $\widehat{L}_k(\zeta^3_k\vert\theta_1)=\epsilon$ & $\widehat{L}_k(\zeta^4_k\vert\theta_1)=\epsilon$\\
\hline
$\theta^{\star}=\theta_2$ & $\widehat{L}_k(\zeta^1_k\vert\theta_2)=p_{k,1}$ & $\widehat{L}_k(\zeta^2_k\vert\theta_2)=1-2\epsilon-p_{k,1}$ & $\widehat{L}_k(\zeta^3_k\vert\theta_2)=\epsilon$ & $\widehat{L}_k(\zeta^4_k\vert\theta_2)=\epsilon$\\
\hline
\end{tabular}
% \end{small}
\end{table*}
\begin{Lem}
\label{Lem_uninformative}
{\bf (Adversaries with uninformative PMFs)}. If every adversary $k\in\mathcal{N}^m$ has uninformative PMFs, then there are no choices of $\widehat{L}_k(\cdot\vert\theta_1),\widehat{L}_k(\cdot\vert\theta_2)$ for which the network is deceived for both $\theta^{\star}=\theta_1$ and $\theta^{\star}=\theta_2$.
\end{Lem}
\begin{proof}
{\em See} Appendix \ref{a2}.
\end{proof}
We proceed by focusing first on the case when there is only one adversary in the network (i.e., $\mathcal{N}^m=\{k\}$) and then extend our results to multiple adversaries. Identifying a set of PMFs $\widehat{L}_k(\cdot\vert\theta_1)$, $\widehat{L}_k(\cdot\vert\theta_2)$ that mislead the network for both $\theta^{\star}=\theta_1$ and $\theta^{\star}=\theta_2$ requires solving the system of inequalities \eqref{lem_ratiosm30} (which is equivalent to \eqref{lem_ratiosm} if $|\mathcal{N}^m|=1$) w.r.t. $\widehat{L}_k(\zeta_k\vert\theta_1),\widehat{L}_k(\zeta_k\vert\theta_2)$, $\zeta_k\in\mathcal{Z}_k$. Since we are in the binary hypothesis setup, the system in \eqref{lem_ratiosm30} is comprised of two inequalities. As we will see it is sufficient to explore a construction of fake likelihood functions that is parametrized by two free variables. The construction is the following. Adversary $k$ selects two distinct realizations of $\boldsymbol{\zeta}_k$, which we denote by $\zeta^1_k,\zeta^2_k\in\mathcal{Z}_k$, without loss of generality. First, it assigns minimum probability mass $\epsilon$ to the remaining realizations of $\boldsymbol{\zeta}_k$, i.e.,
\begin{align}
\label{construction_rest}
    \widehat{L}_k(\zeta_k\vert\theta_1)=\widehat{L}_k(\zeta_k\vert\theta_2)=\epsilon
\end{align}
for all $\zeta_k\neq\zeta^1_k,\zeta^2_k$.

Second, it assigns mass $p_{k,1}$ and $p_{k,2}$ to  $\widehat{L}_k(\zeta^1_k\vert\theta_2)$ and  $\widehat{L}_k(\zeta^2_k\vert\theta_1)$, respectively, i.e.,
\begin{align}
    &\widehat{L}_k(\zeta^1_k\vert\theta_2)=p_{k,1}\\
    &\widehat{L}_k(\zeta^2_k\vert\theta_1)=p_{k,2}.
\end{align}

Since 
$\widehat{L}_k(\zeta_k\vert\theta)$ should sum up to one over $\zeta_k$ for all $\theta\in\Theta$, then from the choices above we have that
\begin{align}
    &\widehat{L}_k(\zeta^1_k\vert\theta_1)=\alpha_k-p_{k,2}\\
    &\widehat{L}_k(\zeta^2_k\vert\theta_2)=\alpha_k-p_{k,1}
\end{align}
where $\alpha_k=1-(|\mathcal{Z}_k|-2)\epsilon$. For the sake of clarity, we present an example of such a construction in Table \ref{construction_table}.

Following the description above, the fake PMFs will have the following form:
\begin{align}
\label{construction10}
    \widehat{L}_k(\zeta_{k}\vert\theta_j)=\begin{cases}p_{k,j'},&\text{ if } \zeta_{k}=\zeta^{j'}_k\\
    \alpha_{k}-p_{k,j'},&\text{ if } \zeta_{k}=\zeta^j_k\\
    \epsilon,&\text{ otherwise }
    \end{cases}
\end{align}
where $j,j'\in\{1,2\}$, $j\neq j'$. 

Next, we investigate when PMFs of the form \eqref{construction10} successfully deceive the network for both $\theta^{\star}=\theta_1$ and $\theta^{\star}=\theta_2$. Replacing \eqref{construction10} into \eqref{lem_ratiosm30} yields
\begin{align}
\label{ineq1}
    &\log\frac{\alpha_{k}-p_{k,1}}{p_{k,2}}>\frac{S_1}{u
    _kL_{k}(\zeta^2_{k}\vert\theta_1)}-\frac{L_{k}(\zeta^1_{k}\vert\theta_1)}{L_{k}(\zeta^2_{k}\vert\theta_1)}\log\frac{p_{k,1}}{\alpha_{k}-p_{k,2}}\\
    \label{ineq2}
    &\log\frac{\alpha_{k}-p_{k,1}}{p_{k,2}}<-\frac{S_2}{u_{k}L_{k}(\zeta^2_{k}\vert\theta_2)}-\frac{L_{k}(\zeta^1_{k}\vert\theta_2)}{L_{k}(\zeta^2_{k}\vert\theta_2)}\log\frac{p_{k,1}}{\alpha_{k}-p_{k,2}}.
\end{align}
Note that \eqref{ineq1} and \eqref{ineq2} correspond to the cases $\theta^{\star}=\theta_1$ and $\theta^{\star}=\theta_2$, respectively. Also note that the other terms appearing in RHS of \eqref{lem_ratiosm30} $L_k(\zeta_k|\theta^{\star})\log\frac{\widehat{L}_k(\zeta_k\vert\theta)}{\widehat{L}_k(\zeta_k\vert\theta^{\star})}$, $\zeta_k\neq\zeta^1_k,\zeta^2_k$ 
vanish due to choice $\widehat{L}_k(\zeta_k\vert\theta_1)=\widehat{L}_k(\zeta_k\vert\theta_2)=\epsilon$. The set of values of $p_{k,1}, p_{k,2}$ that satisfy \eqref{ineq1} and \eqref{ineq2}  define fake PMFs of the form \eqref{construction10} that mislead the network for both $\theta^{\star}=\theta_1$ and $\theta^{\star}=\theta_2$. Before presenting the main result of this section, let us introduce the following quantities:
\begin{align}
&x^-\triangleq\log\frac{\epsilon}{\alpha_{k}-\epsilon},\quad x^+\triangleq\log\frac{\alpha_{k}-\epsilon}{\epsilon}\\
&n_j\triangleq L_{k}(\zeta^j_{k}\vert\theta_2)S_1+L_{k}(\zeta^j_{k}\vert\theta_1)S_2,\quad j=1,2\\
&d_k\triangleq L_{k}(\zeta^2_{k}\vert\theta_2)L_{k}(\zeta^1_{k}\vert\theta_1)-L_{k}(\zeta^2_{k}\vert\theta_1)L_{k}(\zeta^1_{k}\vert\theta_2)\\
\label{int1}
&x_{k,1}'\triangleq\frac{n_2}{u_{k}d_{k}},\quad
x_{k,2}'\triangleq\frac{n_1}{u_{k}d_{k}}.
\end{align}%
Now, we can answer the question of how to appropriately select $p_{k,1},p_{k,2}$ so that construction \eqref{construction10} deceives the network for both  $\theta^{\star}=\theta_1$ and $\theta^{\star}=\theta_2$. 
The following result provides conditions for the existence of such fake likelihood functions as well as a construction of  $\widehat{L}_{k}(\cdot\vert\theta_1),\widehat{L}_{k}(\cdot\vert\theta_2)$, via appropriate selection of $p_{k,1},p_{k,2}$.
\begin{Thm}
\label{region}
    {\bf (Distorted PMFs with known divergences - Single adversary case)}. 
    Let there be only one adversary $k\in\mathcal{N}^m$ (i.e., $|\mathcal{N}^m|=1$) with informative PMFs. Furthermore, let the selected realizations $\zeta^1_{k},\zeta^2_{k}\in\mathcal{Z}_k$ be such that
\begin{align}
\label{det_positive}
L_{k}(\zeta^1_{k}\vert\theta_1)L_{k}(\zeta^2_{k}\vert\theta_2)\neq L_{k}(\zeta^1_{k}\vert\theta_2)L_{k}(\zeta^2_{k}\vert\theta_1).
\end{align}
Let $\epsilon$ satisfy:          
\begin{align}
    \label{e_condition}
    \epsilon<\min\{(e^{|x_{k,1}'|+|\mathcal{Z}_{k}|-1})^{-1},(e^{|x_{k,2}'|+|\mathcal{Z}_{k}|-1})^{-1}\}.
        \end{align}
Then, fake PMFs of the form \eqref{construction10} 
        mislead the network for both $\theta^{\star}=\theta_1$ and $\theta^{\star}=\theta_2$ for the following parameter values.
 \begin{align}
\label{e1e2relation1}
        &p_{k,1}=\frac{e^{x_1}\alpha_k(e^{x_2}-1)}{e^{x_2}-e^{x_1}}\\
    \label{e1e2relation2}
    &p_{k,2}=\frac{\alpha_k(1-e^{x_1})}{e^{x_2}-e^{x_1}}
\end{align}
where $x_1$ is such that $x^+>x_1>x_{k,1}'$ if $d_{k}>0$ and $x^-<x_1<x_{k,1}'$ if $d_{k}<0$ and
\begin{align}
    x_2=\beta_k (x_1-x_{k,1}')+x_{k,2}'
\end{align}
with $|x_2|<x^+$ and $\beta_k$ such that
\begin{align}
    \min_{j\in\{1,2\}}\left\{-\frac{L_{k}(\zeta^1_k\vert\theta_j)}{L_k(\zeta_k^2\vert\theta_j)}\right\}<\beta_{k}<\max_{j\in\{1,2\}}\left\{-\frac{L_{k}(\zeta^1_k\vert\theta_j)}{L_{k}(\zeta^2_k\vert\theta_j)}\right\}.
\end{align}
\end{Thm}
\begin{proof}
{\em See} Appendix \ref{a23}.
\end{proof}
The above result states that even one adversary with informative likelihood functions can construct fake PMFs that mislead the network for $\epsilon$ satisfying \eqref{e_condition}.

The intuition behind the construction presented in this section and the result in Theorem \ref{region} is the following. The system of inequalities \eqref{ineq1}, \eqref{ineq2} is non-linear w.r.t. $p_{k,1},p_{k,2}$. As a result, it is challenging to characterize the region $\widehat{R}_k$ of values of $p_{k,1},p_{k,2}$ that satisfy \eqref{ineq1}, \eqref{ineq2} ({\em see} left sub-figure of Fig. \ref{schematic}). However, we observe that by replacing $\log\frac{\alpha_k-p_{k,1}}{p_{k,2}}$ and $\log\frac{p_{k,1}}{\alpha_k-p_{k,2}}$ with $x_2$ and $x_1$, respectively and by letting $x_1,x_2$ take arbitrary values in $\mathbb{R}$, we obtain a system of inequalities that is linear w.r.t. $x_1,x_2$. Then, it is easy to find the region $\mathcal{R}_k$ of values of $x_1,x_2\in\mathbb{R}$ that satisfy the new system of inequalities ({\em see} right sub-figure of Fig. \ref{schematic}). After solving the new system w.r.t. $x_1,x_2\in\mathbb{R}$, we are able to obtain appropriate $p_{k,1},p_{k,2}$, as described in Theorem \ref{region}, that deceive the network for both $\theta^{\star}=\theta_1$ and $\theta^{\star}=\theta_2$. A schematic representation is given in Fig. \ref{schematic} to enhance intuition behind our approach.
\begin{figure}[!h]
\centering
\includegraphics[width=0.5\textwidth]{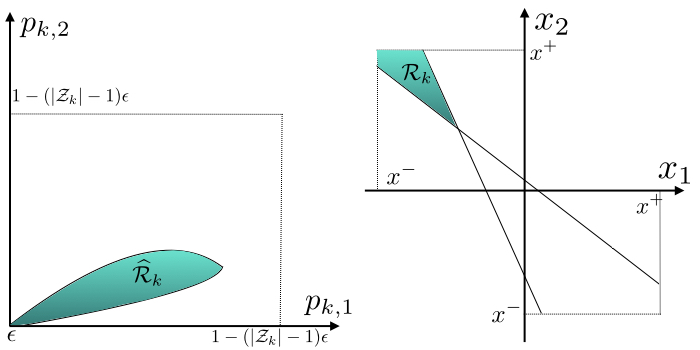}
\caption{Schematic illustration for the construction presented in Theorem \ref{region}}.
\label{schematic}
\end{figure}
\begin{Rem}
\label{r1}
We note that in order to find appropriate $p_{k,1},p_{k,2}$ to construct the fake likelihood functions, knowledge of the normal sub-network divergences $S_1,S_2$ in order to compute $x_{k,1}',x_{k,2}'$ via \eqref{int1} %, \eqref{int2} 
is required. This is the minimal information required by an agnostic adversary that is not aware of the true hypothesis to deceive the network under the type of attacks we consider in this paper.\qedsymb
\end{Rem}
The results from Lemma \ref{Lem_uninformative} and Theorem \ref{region} suggest a way to create distorted likelihood functions that provably mislead the network for both $\theta^{\star}=\theta_1$ and $\theta^{\star}=\theta_2$ for the case where there are more than one adversary in the network (i.e., $|\mathcal{N}^m|>1$). The next result presents such a construction.
\begin{Cor}
\label{multiagent_cor}
    {\bf (Distorted PMFs with known divergences - Multiple adversaries case)}. 
    Let there be at least one adversary with informative PMFs. Then, the following is true. If every adversary $k\in\mathcal{N}^m$ with informative PMFs uses the construction presented in Theorem \ref{region} and every adversary $k\in\mathcal{N}^m$ with uninformative PMFs sets   $\widehat{L}_{k}(\cdot\vert\theta_1)=\widehat{L}_{k}(\cdot\vert\theta_2)=L_{k}(\theta_1)=L_{k}(\theta_2)$ (i.e., does not modify its likelihood functions), then the network is deceived for both $\theta^{\star}=\theta_1$ and $\theta^{\star}=\theta_2$ for $\epsilon$ satisfying
        \begin{align}
        {\begin{small}{
        \label{e_multiagent}
    \epsilon<\min_{k\in\mathcal{N}^{m,+}}\Big\{\min\{(e^{|x_{k,1}'|+|\mathcal{Z}_{k}|-1})^{-1},(e^{|x_{k,2}'|+|\mathcal{Z}_{k}|-1})^{-1}\}\Big\}
    }\end{small}}%
        \end{align}
where $\mathcal{N}^{m,+}$ is the set of adversaries with informative PMFs.
\end{Cor}
\begin{proof}
\textit{See} Appendix \ref{a24}.
\end{proof}
Before concluding this section, we establish a result that will be useful in the sequel.
\begin{Lem} {\bf (Parametrization of fake PMFs with one variable)}. 
\label{one_variable}
PMFs of the form \eqref{construction10} such that $p_{k,1}=p_{k,2}=p$ (parameterization of the fake PMFs with one free variable) are not sufficient to deceive %drive 
the network %to the wrong hypothesis 
for both $\theta^{\star}=\theta_1$ and $\theta^{\star}=\theta_2$, in general. 
\end{Lem}
\begin{proof}
{\em See} Appendix \ref{a25}.
\end{proof}
\subsection{Attack strategies with unknown network divergences}
In general, it is not realistic to assume that knowledge of network divergences $S_1,S_2$ is always available to the adversaries, as it requires access to network topology and normal agents' observation models. Thus, in this section, we investigate what the adversaries can do when they do not know the characteristics of the normal sub-network. Rearranging \eqref{lem_ratiosm}, we define the following cost function.
\begin{align}
    \label{cost_function}
    &\mathcal{C}(\theta^{\star}
    )\triangleq\sum_{k\in\mathcal{N}^n}u_{k}D_{KL}\Bigl(L_k(\theta^{\star}))||L_{k}(\theta)\Bigr)\nonumber\\
    &+\sum_{k\in\mathcal{N}^m}u_{k}\sum_{\zeta_{k}}L_{k}(\zeta_{k}|\theta^{\star})\log\frac{\widehat{L}_{k}(\zeta_{k}|\theta^{\star})}{\widehat{L}_{k}(\zeta_{k}|\theta)}
\end{align}
where $\theta^{\star},\theta\in\Theta, \theta^{\star}\neq\theta$. Then, condition \eqref{lem_ratiosm} is equivalent to
\begin{align}
\label{eq_5equiv}
\mathcal{C}(\theta^{\star})<0.    
\end{align}
We observe that the second term in \eqref{cost_function} is under malicious agents' control. Thus, one option for the adversaries is to minimize \eqref{cost_function} over $\widehat{L}_k(\cdot\vert\theta_1),\widehat{L}_k(\cdot\vert\theta_2)$ to increase the chances that \eqref{eq_5equiv} is satisfied. However, $\theta^{\star}$ is unknown as well. A viable alternative is to treat the true state $\theta^{\star}$ as a random variable, i.e., $\boldsymbol{\theta}^{\star}$. We assume that adversaries share a common prior over the states and since no evidence about the true state is available beforehand, we set $\mathbb{P}(\boldsymbol{\theta}^\star=\theta_1)=\mathbb{P}(\boldsymbol{\theta}^\star=\theta_2)=1/2$. 
Thus, taking expectation over the true state $\boldsymbol{\theta}^{\star}$ in \eqref{cost_function} leads to the following minimization problem for the malicious agents:
\begin{align}
    \label{opt_altao}
    &\min_{\widehat{L}_{k}(\cdot\vert\theta_1),\widehat{L}_{k}(\cdot\vert\theta_2)}\frac{1}{2}\Big(\mathcal{C}(\theta_1
    )+\mathcal{C}(\theta_2
    )\Big),\,\,\ k\in\mathcal{N}^m
    \\
    &\text{s.t.}\,\,\, \widehat{L}_{k}(\zeta|\theta)\geq\epsilon,\,\,\,\quad\quad\quad\forall\zeta\in\mathcal{Z}_{k},\theta\in\Theta,
    \nonumber\\
    &\,\,\,\,\,\,\,\, \sum_{\zeta\in\mathcal{Z}_{k}}\widehat{L}_{k}(\zeta|\theta)=1,\quad\,\,\,\forall \theta\in\Theta
    \nonumber
\end{align}
The first constraint in the optimization problem is due to Assumption \ref{finiteness} and the second one is due to the fact that $\widehat{L}_k(\zeta_k\vert\theta)$ should sum up to one over $\zeta_k$ for every $\theta\in\Theta$. It should be noted that the solution to the minimization problem above, denoted by  $\widehat{L}^{\star}_{k}(\cdot\vert\theta_1),\widehat{L}^{\star}_{k}(\cdot\vert\theta_2)$, minimizes the average of \eqref{cost_function} for $\theta^{\star}=\theta_1$ and $\theta^{\star}=\theta_2$. 
This means that \eqref{lem_ratiosm} is not necessarily satisfied for both $\theta^{\star}=\theta_1$ and $\theta^{\star}=\theta_2$ if adversaries utilize $\widehat{L}^{\star}_{k}(\cdot\vert\theta_1),\widehat{L}^{\star}_{k}(\cdot\vert\theta_2)$, but adversaries try to satisfy \eqref{lem_ratiosm} on average. 

The optimization problem decomposes across agents $k\in\mathcal{N}^m$  
and thus, \eqref{opt_altao} reduces to the following for each agent $k\in\mathcal{N}^m$:
\begin{align}
% {\begin{small}{
\label{opt_prao2pr}
    &
    \hspace{-1mm}\min_{\widehat{L}_{k}(\cdot\vert\theta_1)}\sum_{\zeta\in\mathcal{Z}_k}\hspace{-0.9mm}Z_{k}(\zeta)\hspace{-0.5mm}\log\widehat{L}_{k}(\zeta|\theta_1)\hspace{-1mm}-\hspace{-1.5mm}\max_{\widehat{L}_{k}(\cdot\vert\theta_2)}\sum_{\zeta\in\mathcal{Z}_k}\hspace{-0.9mm}Z_{k}(\zeta)\hspace{-0.5mm}\log\widehat{L}_{k}(\zeta|\theta_2)\\
    &
    \text{s.t.}\,\,\,\,\,\,\, \widehat{L}_{k}(\zeta|\theta_1)\geq\epsilon,\quad\quad\quad\quad
    \widehat{L}_{k}(\zeta|\theta_2)\geq\epsilon,\quad\forall\zeta\in\mathcal{Z}_{k}\nonumber\\
    &\quad\,\,\,\,\,\sum_{\zeta\in\mathcal{Z}_{k}}\widehat{L}_{k}(\zeta|\theta_1)=1,\quad\,\,\sum_{\zeta\in\mathcal{Z}_{k}}\widehat{L}_{k}(\zeta|\theta_2)=1\nonumber
% }\end{small}}
\end{align}
where we introduced:
\begin{align}
\label{eq_confidence}
Z_{k}(\zeta)\triangleq L_{k}(\zeta|\theta_1)-L_{k}(\zeta|\theta_2),\, \quad\zeta\in\mathcal{Z}_{k}.
\end{align}
Note that each coefficient $Z_{k}(\zeta)$ expresses a measure of {\em confidence} that an observation $\zeta$ resulted from state $\theta_1$ instead of $\theta_2$. If $Z_k(\zeta)$ is positive, then $\zeta$ is more likely to have been generated by state $\theta_1$ instead of state $\theta_2$, while if $Z_k(\zeta)$ is negative, then $\zeta$ is more likely to have been generated by state $\theta_2$ instead of $\theta_1$. Let us define the set
\begin{align}
% \begin{small}
\mathcal{D}^1_{k}=\{\zeta\in\mathcal{Z}_k:Z_k(\zeta)
\geq0,\quad k\in\mathcal{N}^m\}
% \end{small}%
\end{align}
which is comprised of those observations $\zeta\in\mathcal{Z}_{k}$ that are more (or equally) likely that they have been generated by state $\theta_1$ instead of state $\theta_2$. Respectively, the set $\mathcal{D}^2_{k}=\mathcal{Z}_{k}\setminus\mathcal{D}^1_{k},k\in\mathcal{N}^m$ is comprised of the observations that are more likely to have been generated by state $\theta_2$ instead of $\theta_1$. Before presenting the solution of the optimization problem, we establish a useful result for the sets $\mathcal{D}^1_{k}$, $\mathcal{D}^2_{k}$.

\begin{Lem} {\bf (Non-empty Partition)}. 
\label{uniform_rem}
The sets $\mathcal{D}^1_{k}$ and $\mathcal{D}^2_{k}$ are both non-empty for an adversary $k\in\mathcal{N}^m$ with informative PMFs.
\end{Lem}
\begin{proof}
{\em See} Appendix \ref{a3}.
\end{proof}
It can be seen from \eqref{eq_confidence} that $Z_k(\zeta)=0$ for all $\zeta\in\mathcal{Z}_k$ for an adversary $k\in\mathcal{N}^m$ with uninformative PMFs and as a result the objective function of the optimization problem \eqref{opt_prao2pr} is equal to $0$ for any choice of $\widehat{L}_k(\cdot\vert\theta_1)$, $\widehat{L}_k(\cdot\vert\theta_2)$. Thus, in the remainder of the paper we assume that all adversaries have informative PMFs. The solution to \eqref{opt_prao2pr} is given by the following result. 
\begin{Thm}
\label{opt_attackth}
{\bf (Distorted PMFs with unknown divergences)}. The attack strategy optimizing \eqref{opt_prao2pr} for an adversary $k\in\mathcal{N}^m$ is given by
% \begin{small}
\begin{align}
    \label{opt_attack}
    &\widehat{L}_{k}(\zeta|\theta_j)=\begin{cases} \epsilon,\quad\,\text{if } \zeta\in\mathcal{D}^j_{k},\vspace{1mm}\\
                            \displaystyle
                    \frac{Z_{k}(\zeta)(1-|\mathcal{D}^j_{k}|\epsilon)}{\sum\limits_{\zeta\notin\mathcal{D}^j_{k}}Z_{k}(\zeta)},\quad\,\text{if } \zeta\notin\mathcal{D}^j_{k}
                            \end{cases}
\end{align}
% \end{small}
where $j\in\{1,2\}$.%
\end{Thm}
\begin{proof}
{\em See} Appendix \ref{a32}.
\end{proof}
A graphical representation of the solution \eqref{opt_attack} %-\eqref{opt_attack_theta_2} 
is presented in Fig. \ref{equal_priors}. The intuition behind the attack strategy is the following. We focus on the construction for $\widehat{L}_{k}(\cdot\vert\theta_1)$ and the rationale is the same for $\widehat{L}_{k}(\cdot\vert\theta_2)$. The constructed PMF $\widehat{L}_{k}(\cdot\vert\theta_1)$ is such that the least possible probability mass (i.e., $\epsilon$) is assigned to every observation $\zeta$ that is more likely to have been generated from state $\theta_1$ (i.e., for all $\zeta\in\mathcal{D}^1_{k}$). 
For the remaining observations that are more likely to be generated from $\theta_2$ (i.e., $\zeta\in\mathcal{D}^2_{k}$) 
the probability mass placed on every $\zeta\in\mathcal{D}^2_{k}$ is in proportion to the difference in probability that $\zeta$ is generated from $\theta_2$ instead of $\theta_1$. 
The more likely it is for $\zeta$ to be generated from $\theta_2$, the more probability mass is placed on $\widehat{L}_k(\zeta|\theta_1)$. 
\begin{Rem}
\label{intuition_str}
Intuitively, the above strategy indicates that 
the fake PMFs should be constructed by following the rationale to 
``\textit{inflate confidence towards the least likely state}". \qedsymb
\end{Rem}

\begin{figure}[!h]
\centering
\includegraphics[width=0.5\textwidth]{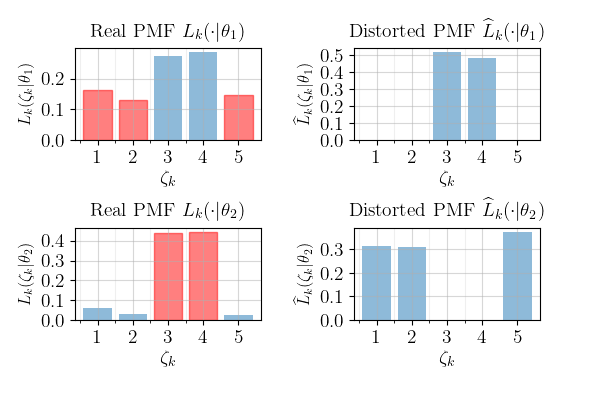}
\caption{An example of actual and distorted PMFs given by Theorem \ref{opt_attackth} with $|\mathcal{Z}_{k}|=5$. Red color depicts the higher value of $L_{k}(\zeta_{k}\vert\theta)$ for every observation $\zeta_{k}$ w.r.t. states (i.e., $L_{k}(\zeta_{k}\vert\theta)$ in red are such that $L_{k}(\zeta_{k}\vert\theta)>L_{k}(\zeta_{k}\vert\theta')$, $\theta\neq\theta'$). We set $\epsilon=10^{-3}$.}
\label{equal_priors}
\end{figure}%
\section{Analysis of the approximate solution}
An interesting question that arises is when the approximate solution presented in Theorem \ref{opt_attackth} misleads the network for both $\theta^{\star}=\theta_1$ and $\theta^{\star}=\theta_2$. First, we examine the case where there is only one adversary in the network (i.e., $|\mathcal{N}^m|=1$). By replacing the solution presented in Theorem \ref{opt_attackth} into \eqref{lem_ratiosm30} (which is equivalent to \eqref{lem_ratiosm} for $|\mathcal{N}^m|=1$) we have
% \normalsize
% {\begin{small}{
\begin{align}
\label{app_c1}
    &\log\frac{1-|\mathcal{D}^2_{k}|\epsilon}{\epsilon}\hspace{-1mm}>\hspace{-1mm}\frac{S_1-u_{k}(c_{k,1}+b_{k,1})}{u_{k}\xi_{k,1}}-\frac{\sigma_{k,1}}{\xi_{k,1}}\log\frac{\epsilon}{1-|\mathcal{D}^1_{k}|\epsilon}\\
    \label{app_c2}
    &\log\frac{1-|\mathcal{D}^2_{k}|\epsilon}{\epsilon}\hspace{-1mm}<\hspace{-1mm}-\frac{S_2+u_{k}(c_{k,2}+b_{k,2})}{u_{k}\xi_{k,2}}-\frac{\sigma_{k,2}}{\xi_{k,2}}\log\frac{\epsilon}{1-|\mathcal{D}^1_{k}|\epsilon}
\end{align}
% }\end{small}}%
% \normalsize
where
\begin{align}
% {\begin{small}{
\label{approx_c1}
&c_{k,j}=\sum_{\zeta\in\mathcal{D}^1_{k}}L_{k}(\zeta\vert\theta_j)\log\frac{Z_{k}(\zeta)}{\sum_{\zeta'\in\mathcal{D}^1_{k}}Z_{k}(\zeta')}\\
\label{approx_c2}
&b_{k,j}=\sum_{\zeta\in\mathcal{D}^2_{k}}L_{k}(\zeta\vert\theta_j)\log\frac{Z_{k}(\zeta)}{\sum_{\zeta'\in\mathcal{D}^2_{k}}Z_{k}(\zeta')}\\
\label{approx_c3}
&\xi_{k,j}=\sum_{\zeta\in\mathcal{D}^1_{k}}L_{k}(\zeta\vert\theta_j)\\
\label{approx_c4}
&\sigma_{k,j}=\sum_{\zeta\in\mathcal{D}^2_{k}}L_{k}(\zeta\vert\theta_j)
% }\end{small}}
\end{align}
$j\in\{1,2\}$.

The strategy presented in Theorem \ref{opt_attackth} does not mislead the network for both $\theta^{\star}=\theta_1$ and $\theta^{\star}=\theta_2$, in general. 
It can be easily seen by counterexample. Let $\mathcal{N}^m=\{k\}$ with $\mathcal{Z}_{k}=\{\zeta^1,\zeta^2\}$. Also, let $L_k(\zeta^1\vert\theta_1)<L_k(\zeta^1\vert\theta_2)$, which implies that  $L_k(\zeta^2\vert\theta_1)>L_k(\zeta^2\vert\theta_2)$, since $|\mathcal{Z}_k|=2$. From Lemma \ref{uniform_rem} we have that $|\mathcal{D}^1_{k}|=|\mathcal{D}^2_{k}|=1$. Then, \eqref{app_c1}, \eqref{app_c2} yield
\begin{align}
\label{counter1}
    &\log\frac{1-\epsilon}{\epsilon}>\frac{S_1}{u_{k}L_k(\zeta^2\vert\theta_1)}-\frac{L_k(\zeta^1\vert\theta_1)}{L_k(\zeta^2\vert\theta_1)}\log\frac{\epsilon}{1-\epsilon}\\
    \label{counter2}
    &\log\frac{1-\epsilon}{\epsilon}<-\frac{S_2}{u_{k}L_k(\zeta^2\vert\theta_2)}-\frac{L_k(\zeta^1\vert\theta_2)}{L_k(\zeta^2\vert\theta_2)}\log\frac{\epsilon}{1-\epsilon}
\end{align}
By treating $\epsilon$ as a free variable, we observe that the above system of inequalities is of the same form as the system we get under the construction with one free variable in Lemma \ref{one_variable} with $p_{k,1}=p_{k,2}=\epsilon$. Then, it follows from Lemma \ref{one_variable} 
that the system of inequalities \eqref{counter1}, \eqref{counter2} may not have a solution for any $\epsilon$.

However, under certain conditions the network is misled for sufficiently small $\epsilon$ for any $\theta^{\star}\in\Theta$, meaning that there exists $\epsilon^{\star}$ such that  \eqref{lem_ratiosm30} is satisfied for every $j\in\{1,2\}$ and every $0<\epsilon<\epsilon^{\star}$ if the adversary follows the attack strategy given by Theorem \ref{opt_attackth}. We identify such cases in the sequel. First, we introduce a useful definition.
\begin{Def}
\label{separation_definition}
 {\bf(Separable observations)}. Given a partition $\mathcal{Z}^1,\mathcal{Z}^2$ of the set of observations $\mathcal{Z}_{k}$ of an agent $k\in\mathcal{N}$, agent's $k$ observations are called {\em  separable} if the following is true:
\begin{align}
&\sum_{\zeta\in\mathcal{Z}^j}L_{k}(\zeta\vert\theta_1)>\sum_{\zeta\in\mathcal{Z}^{j'}}L_{k}(\zeta\vert\theta_1)\\
&\sum_{\zeta\in\mathcal{Z}^j}L_{k}(\zeta\vert\theta_2)<\sum_{\zeta\in\mathcal{Z}^{j'}}L_{k}(\zeta\vert\theta_2)
\end{align}
for some $j,'j\in\{1,2\}$ such that $j\neq j'$, which is equivalent to
\begin{align}
    &\sum_{\zeta\in\mathcal{Z}^j}L_{k}(\zeta\vert\theta_1)>0.5\\
&\sum_{\zeta\in\mathcal{Z}^j}L_{k}(\zeta\vert\theta_2)<0.5
\end{align}
for some $j\in\{1,2\}$. 
Otherwise, agent $k$'s observations are called {\em non-separable}.\qedsymb
\end{Def}
A simple example of separable and non-separable observations of an agent $k$ for the case of binary observation space (i.e., $|\mathcal{Z}_k|=2$) is given in Fig. \ref{sep_obs_example_fig}. As illustrated in the upper row, if the observations are separable, then  one observation is more likely to be generated by one state, while the other observation is more likely to be generated by the other state (the observations that are most likely for a given state are depicted in red color). In contrary, for the case of non-separable observations (lower row), we observe that one observation is more likely to be generated by both states ($\zeta_2$ in the example). The rationale extends in a straightforward way to the case of multiple observations where the single observations are replaced by $\sigma_{k,j}$ and $\xi_{k,j}$, $j=1,2$ given by \eqref{approx_c3} and \eqref{approx_c4}, respectively.
\begin{figure}[!h]
\centering
\includegraphics[width=0.4\textwidth]{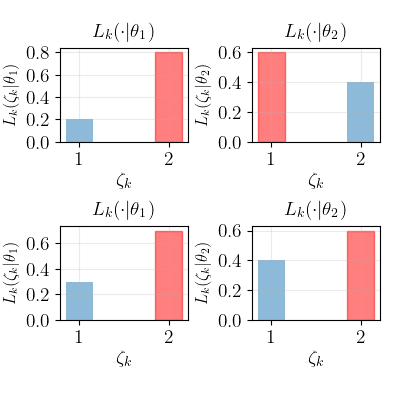}  
\vspace{-5mm}
\caption{An example of separable (upper row) and non-separable observations (lower row) for the case of binary observation space (i.e., $|\mathcal{Z}_k|=2$).}
\label{sep_obs_example_fig}
\end{figure}
\begin{Thm}
\label{efficiciency_Th3} {\bf(Global deception under strategy with unknown divergences. Single adversary case)}. 
Let $\mathcal{N}^m=\{k\}$. Then, under the attack strategy presented in Theorem \ref{opt_attackth} the network is misled for both $\theta^{\star}=\theta_1$ and $\theta^{\star}=\theta_2$ for sufficiently small $\epsilon$ if agent $k$'s observations are separable under the partition $\mathcal{D}^1_{k},\mathcal{D}^2_{k}$, which is equivalent to the following conditions.
\begin{align}
\label{condition_Th31}
    &\sigma_{k,1}<\xi_{k,1}\\
    \label{condition_Th32}
    &\sigma_{k,2}>\xi_{k,2}.
\end{align}
\end{Thm}
\begin{proof}
{\em See} Appendix \ref{a33}.
\end{proof}
\noindent The result easily extends to the multiple adversaries case.
\begin{Cor}{\bf(Global deception under strategy with unknown divergences. Multiple adversaries case)}. 
Let all adversaries follow the attack strategy presented in Theorem \ref{opt_attackth}. Then, the network is misled for both $\theta^{\star}=\theta_1$ and $\theta^{\star}=\theta_2$ for sufficiently small $\epsilon$ if for all adversaries $k\in\mathcal{N}^m$ agent $k$'s observations are separable under the partition $\mathcal{D}^1_{k},\mathcal{D}^2_{k}$.
\end{Cor}
\begin{proof}
The result follows from Theorem \ref{efficiciency_Th3}. Since for sufficiently small $\epsilon$, conditions \eqref{app_c1}, \eqref{app_c2} are satisfied for every $k\in\mathcal{N}^m$, or equivalently \eqref{lem_ratiosm30} is satisfied for $j=1,2$ for every $k\in\mathcal{N}^m$, then \eqref{inequality_re}, or equivalently \eqref{lem_ratiosm}, is satisfied as well (\textit{see} Remark \ref{suf_condition}) for sufficiently small $\epsilon$, which implies that the network is misled for both $\theta^{\star}=\theta_1$ and $\theta^{\star}=\theta_2$.
\end{proof}
Conditions \eqref{condition_Th31}, \eqref{condition_Th32} express that if the structure of the observation model of an adversary is such that the observations are partitioned into $\mathcal{D}^1_{k},\mathcal{D}^2_{k}$ in such a way that the probability of appearance of an observation that belongs in $\mathcal{D}^1_{k}$ is greater than the probability of appearance of an observation that belongs in $\mathcal{D}^2_{k}$ given $\theta^{\star}=\theta_1$ and smaller given $\theta^{\star}=\theta_2$, then the network is misled for sufficiently small $\epsilon$. The intuition behind this result is the following. We observe that the strategy  given by Theorem \ref{opt_attackth} dictates to  
\textit{inflate} the confidence that the state that generated the particular observation is the least likely one (according to the adversary's true likelihood functions) ({\em see} Remark \ref{intuition_str}). Thus, if the true likelihood functions of the adversary are such that the confidence provided by the generated observations satisfy the separation principle described in Definition \ref{separation_definition}, then the confidence placed in the most likely events can be sufficiently small and the strategy will mislead the network for both $\theta^{\star}=\theta_1$ and $\theta^{\star}=\theta_2$. To make this observation clear we present the following example.
\begin{Exa}
\label{example_ref}
Let $\mathcal{N}^m=\{k\}$. Adversary $k$'s observation model is given by Table \ref{table_example}.
\begin{table}[!h]
\renewcommand{\arraystretch}{1.3}
\caption{Adversary $k$'s observation matrix.}
\label{table_example}
\centering
\begin{tabular}{|c||c|c|}
\hline
&$\zeta^1$&$\zeta^2$\\
\hline
$\theta^{\star}=\theta_1$ & $L_{k}(\zeta^1\vert\theta_1)$ & $L_{k}(\zeta^2\vert\theta_1)$\\
\hline
$\theta^{\star}=\theta_2$ & $L_{k}(\zeta^1\vert\theta_2)$ & $L_{k}(\zeta^2\vert\theta_2)$\\
\hline
\end{tabular}
\end{table}
Also, let $\mathcal{D}^1_{k}=\{\zeta^1\},\mathcal{D}^2_{k}=\{\zeta^2\}$. Then, \eqref{approx_c1}-\eqref{approx_c4} yield $\sigma_{k,1}=L_{k}(\zeta^1\vert\theta_1),\sigma_{k,2}=L_{k}(\zeta^1\vert\theta_2),\xi_{k,1}=L_{k}(\zeta^2\vert\theta_1),\xi_{k,2}=L_{k}(\zeta^2\vert\theta_2)$. Eqs. \eqref{condition_Th31}, \eqref{condition_Th32}  in this example are equivalent to
% \begin{small}
\begin{align}
\label{exa1}
&L_{k}(\zeta^1\vert\theta_1)<L_{k}(\zeta^2\vert\theta_1)   \\
\label{exa2}
&L_{k}(\zeta^1\vert\theta_2)>L_{k}(\zeta^2\vert\theta_2).
\end{align}
% \end{small}%
This means that if the true state is $\theta_1$, then observation $\zeta^1$ is less likely than $\zeta^2$, while if the true state is $\theta_2$, then observation $\zeta^1$ is more likely than $\zeta^2$. The fake likelihoods given by Theorem \ref{opt_attackth} are shown in Table \ref{table_fake}.
\begin{table}[!h]
\renewcommand{\arraystretch}{1.3}
\caption{Fake likelihood functions.}
\label{table_fake}
\centering
\begin{tabular}{|c||c|c|}
\hline
&$\zeta^1$&$\zeta^2$\\
\hline
$\theta^{\star}=\theta_1$ & $\widehat{L}_{k}(\zeta^1\vert\theta_1)=1-\epsilon$ & $\widehat{L}_{k}(\zeta^2\vert\theta_1)=\epsilon$\\
\hline
$\theta^{\star}=\theta_2$ & $\widehat{L}_{k}(\zeta^1\vert\theta_2)=\epsilon$ & $\widehat{L}_{k}(\zeta^2\vert\theta_2)=1-\epsilon$\\
\hline
\end{tabular}
\end{table}
By treating $\epsilon$ as a free variable and by 
replacing $\log\frac{\epsilon}{1-|\mathcal{D}^1_{k}|\epsilon}$ and $\log\frac{1-|\mathcal{D}^2_{k}|\epsilon}{\epsilon}$ with $\widetilde{x}_1$ and $\widetilde{x}_2$, respectively, in \eqref{app_c1}, \eqref{app_c2} we have
\begin{align}
\label{exa_c1}
    &\widetilde{x}_2>\frac{S_1}{u_{k}L_{k}(\zeta^2\vert\theta_1)}-\frac{L_{k}(\zeta^1\vert\theta_1)}{L_{k}(\zeta^2\vert\theta_1)}\widetilde{x}_1\\
    \label{exa_c2}
    &\widetilde{x}_2<-\frac{S_2}{u_{k}L_{k}(\zeta^2\vert\theta_2)}-\frac{L_{k}(\zeta^1\vert\theta_2)}{L_{k}(\zeta^2\vert\theta_2)}\widetilde{x}_1.
\end{align}
Let $\widetilde{x}_1,\widetilde{x}_2$ take arbitrary values in $\mathbb{R}$. Further, let $\mathcal{R}'$ be the set of values of $\widetilde{x}_1,\widetilde{x}_2\in\mathbb{R}$ that satisfy the inequalities above. 
However, $\widetilde{x}_1,\widetilde{x}_2$ cannot take arbitrary values and we have $\widetilde{x}_2=\log\frac{1-\epsilon}{\epsilon}$ and $\widetilde{x}_1=\log\frac{\epsilon}{1-\epsilon}$, which implies that
\begin{align}
    \widetilde{x}_2=-\widetilde{x}_1.
\end{align}
We observe that due to \eqref{exa1}, \eqref{exa2}, the line $\widetilde{x}_2=-\widetilde{x}_1$ is always contained in the region $\mathcal{R}'$ for sufficiently small $\widetilde{x}_1$, meaning for sufficiently small $\epsilon$. This means that \eqref{lem_ratiosm30} is satisfied for sufficiently small $\epsilon$, which implies that the network is deceived for both $\theta^{\star}=\theta_1$ and $\theta^{\star}=\theta_2$.\qedsymb
\end{Exa}
\begin{Rem}
From the above example, we can draw a connection between the two strategies presented in this paper (Theorem \ref{region} and Theorem \ref{opt_attackth}). First, let us note that the construction in Theorem \ref{region} utilizes two realizations of $\mathcal{Z}_k$, while the construction in Theorem \ref{opt_attackth} utilizes every $\zeta_k\in\mathcal{Z}_k$. That is why the comparison is possible in the example presented above, where we have $|\mathcal{Z}_k|=2$. So, the strategy given by Theorem \ref{opt_attackth} in Example \ref{example_ref} corresponds to the values $p_{k,1}=\epsilon,p_{k,2}=\epsilon$ of the construction presented in Theorem \ref{region}, which yields $x_1=x^-$, $x_2=x^+$ in the transformed linear domain (for convenience also {\em see} Fig. \ref{schematic}). Intuitively, while Theorem \ref{region} characterizes the set of all values of $p_{k,1},p_{k,2}$ that deceive the network for both states under the considered construction, the result of Theorem \ref{opt_attackth} (applied in the binary observation space example) yields the extreme value $p_{k,1}=p_{k,2}=\epsilon$ (or $p_{k,1}=p_{k,2}=1-\epsilon$). Depending on the value of $\epsilon$, as well as on whether the agent's observations are separable, this extreme point might belong in the set of values that mislead the network or not.
\qedsymb
\end{Rem}

In the following Section, we present simulation results that demonstrate the impact of adversaries on the learning process of the network under the attacks studied in this paper.
\section{Simulations}
We assume $15$ agents, with $11$ normal and $4$ malicious agents. Each agent assigns uniform combination weights to its neighbors. The network (illustrated in Fig. \ref{network}) is strongly connected and randomly generated. For simplicity, we assume that all agents have the same observation model, meaning $L_{k}(\theta)=L(\theta)$, for every $k\in\mathcal{N}^m$ and for every $\theta\in\Theta$.

\vspace{-5mm}
\begin{figure}[!h]
\centering
\includegraphics[width=0.4\textwidth]{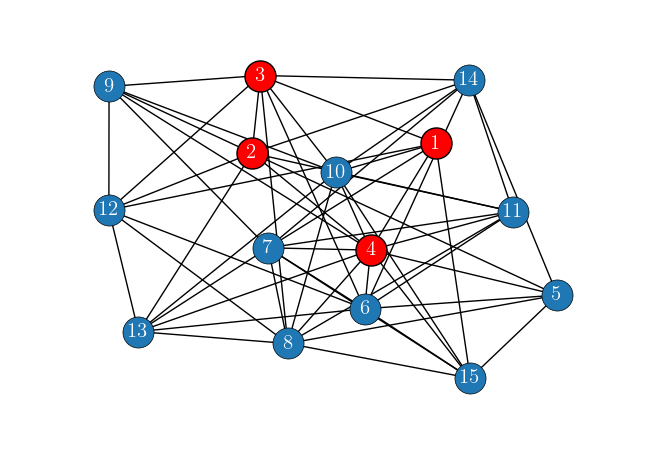}  
\vspace{-5mm}
\caption{Network topology. Adversaries are illustrated in red color, normal agents are illustrated in blue color. All agents have self-loops (not depicted in the figure), meaning $a_{kk}>0$ for all $k\in\mathcal{N}$.}
\label{network}
\end{figure}
\subsection{The impact of observation models and network topology}
First, we want to highlight the impact of the \textit{informativeness} of agents' observation models, which is expressed via agents' KL divergences, and network topology, which is expressed via agents' centrality, on the learning performance under adversarial strategies. For ease of exposition, we assume that all agents observe the state through a binary symmetric channel (BSC), meaning $\mathcal{Z}_k=\{\zeta^1,\zeta^2\}$ for all $k\in\mathcal{N}$, with observation probabilities $L(\zeta^1|\theta_1)=L(\zeta^2|\theta_2)=p$ and $L(\zeta^2|\theta_1)=L(\zeta^1|\theta_2)=1-p$. Finally, we assume that the true hypothesis is $\theta^{\star}=\theta_1$. In this subsection, we illustrate only the impact of the attack strategy with unknown divergences given by Theorem \ref{opt_attackth}. The impact of the strategy with known divergences given by Theorem \ref{region} is demonstrated in the next subsection.

In Fig. \ref{network_ld} the BSC is parametrized with $p=0.8$, while in Fig. \ref{network_d} agents' observation models are more discriminating between the two states (and thus more informative) with $p=0.9$. To show these dependencies we restrict the admissible distortion level of the fake likelihoods by setting $\epsilon=5\times10^{-3}.$ 
 Apart from the dependence on the observation models, we demonstrate the impact of network topology by considering the random topology presented in Fig. \ref{network} and the star topology. In the star topology, the central agent is malicious, resulting in higher overall centrality of the adversaries compared to the random network topology case. 

In both Figs. \ref{network_ld}, \ref{network_d} the agents' belief evolution is shown. We utilize a granular light green to dark orange color map to show each agent's belief $\boldsymbol{\mu}_{k,i}(\theta^{\star})$. Agents in light green denote agents whose beliefs are close to the wrong state, dark orange denotes agents whose beliefs are close to the true state and light orange denotes agents whose beliefs are close to uniform distribution (i.e., $\boldsymbol{\mu}_{k,i}=[0.5,0.5]^{\mathsf{T}}$). As we observe in Fig. \ref{network_ld}, the network converges to the wrong hypothesis at steady-state ({\em see} last row) under the attack strategy given by Theorem \ref{opt_attackth} in both cases of random network (depicted in Fig. \ref{network}) and star network. 
\begin{figure}[!h]
\centering\hspace{-5mm}
\includegraphics[width=0.5\textwidth]{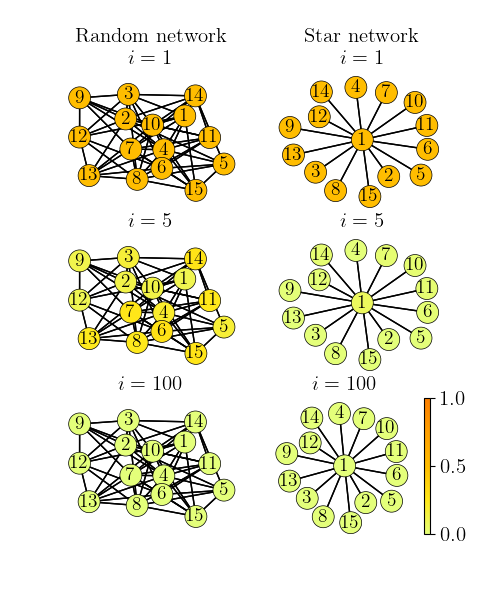}%
\vspace{-3mm}
\caption{Agents' belief evolution (agents' observations are given by a BSC with $p=0.8$) in time instants $i=0$, $i=4$ and $i=100$. Dark Orange: Beliefs close to the true state. Light green: Beliefs close to the wrong state. Light Orange: Beliefs close to uniform distribution. Left sub-figures: Random topology ({\em see} Fig. \ref{network}), Right sub-figures: Star topology.  
Adversaries are agents $1,2,3,4$ ({\em see} Fig. \ref{network}) and they use the attack strategy with unknown divergences given by Theorem \ref{opt_attackth}.}
\label{network_ld}
\end{figure}
\begin{figure}[!h]
\centering\hspace{-5mm}
\includegraphics[width=0.5\textwidth]{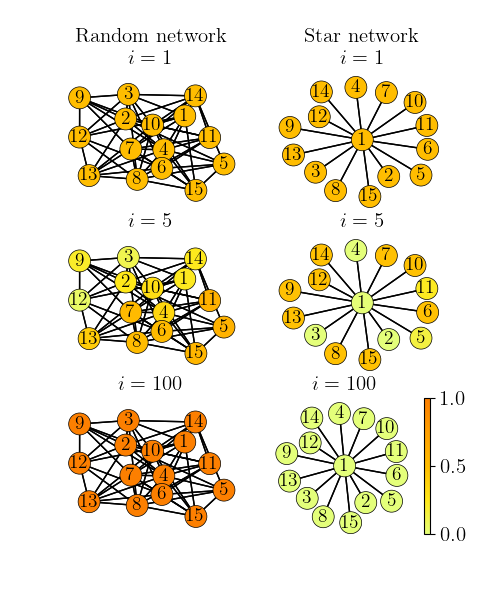}%
\vspace{-3mm}
\caption{Agents' belief evolution with highly discriminative models (agents' observations are given by a BSC with $p=0.9$) at time instants $i=0$, $i=4$ and $i=100$. Dark Orange: Beliefs close to the true state. Light green: Beliefs close to the wrong state. Light Orange: Beliefs close to uniform distribution. Left sub-figures: Random topology ({\em see} Fig. \ref{network}), Right sub-figures: Star topology. Adversaries are agents $1,2,3,4$ ({\em see} Fig. \ref{network}) and they use the attack strategy with unknown divergences given by Theorem \ref{opt_attackth}.}
\label{network_d}
\end{figure}

The same rationale is followed in the experiments conducted for more discriminating models ($p=0.9$). As we observe in Fig. \ref{network_d}, the impact of malicious behavior using the strategy presented in Theorem \ref{opt_attackth} is smaller in this setup, since normal agents are more capable to discriminate between the two hypotheses. More specifically, in the left sub-figure in the last row of Fig. \ref{network_d}, we see that the network converges to the true state for the random network topology. On the other hand, for the star topology, where the central agent is malicious, the network is misled under the same attack strategy, as presented in the right sub-figure in the last row. This is because the overall centrality of the malicious agents is bigger in star topology compared to the random network topology shown in Fig. \ref{network}.

The impact of the observation models and agents' centrality is explored in Fig \ref{obs_cen}. We plot the average of agents' beliefs for the true state (i.e., $\bar{\boldsymbol{\mu}}_i(\theta^{\star})\triangleq\frac{\sum_{k\in\mathcal{N}}\boldsymbol{\mu}_{k,i}(\theta^{\star})}{|\mathcal{N}|}$) at steady-state when adversaries use the strategy with unknown divergences (given by Theorem \ref{opt_attackth}) for different values of the BSC probability $p$ in the left sub-figure and for different values of overall adversaries' centrality (i.e., $\sum_{k\in\mathcal{N}^m}u_k$) in the right sub-figure. 
We consider the network topology given by Fig. \ref{network}. We observe that there is a phase transition phenomenon. More specifically, as we see in the left sub-figure, if agents' observation models are sufficiently discriminating between the two states (i.e., the BSC probability $p$ is close to $1$), then network beliefs converge to the true state, while as the observation models become less discriminating (the BSC probability $p$ is close to $0.5$) adversaries drive the network to the wrong state. Likewise, in the right sub-figure we see that if the overall adversaries' centrality is small, then the network correctly identifies the true hypothesis, whereas if adversaries' centrality is sufficiently high, then the network is misled. Finally, we note that the critical values of BSC probability $p$ and overall adversaries' centrality where the phase transition occurs match the values predicted by Theory (by solving \eqref{lem_ratiosm} w.r.t. $p$ and $\sum_{k\in\mathcal{N}^m}u_{k}$, respectively).
\begin{figure}[!h]
\centering
\includegraphics[width=1\linewidth]{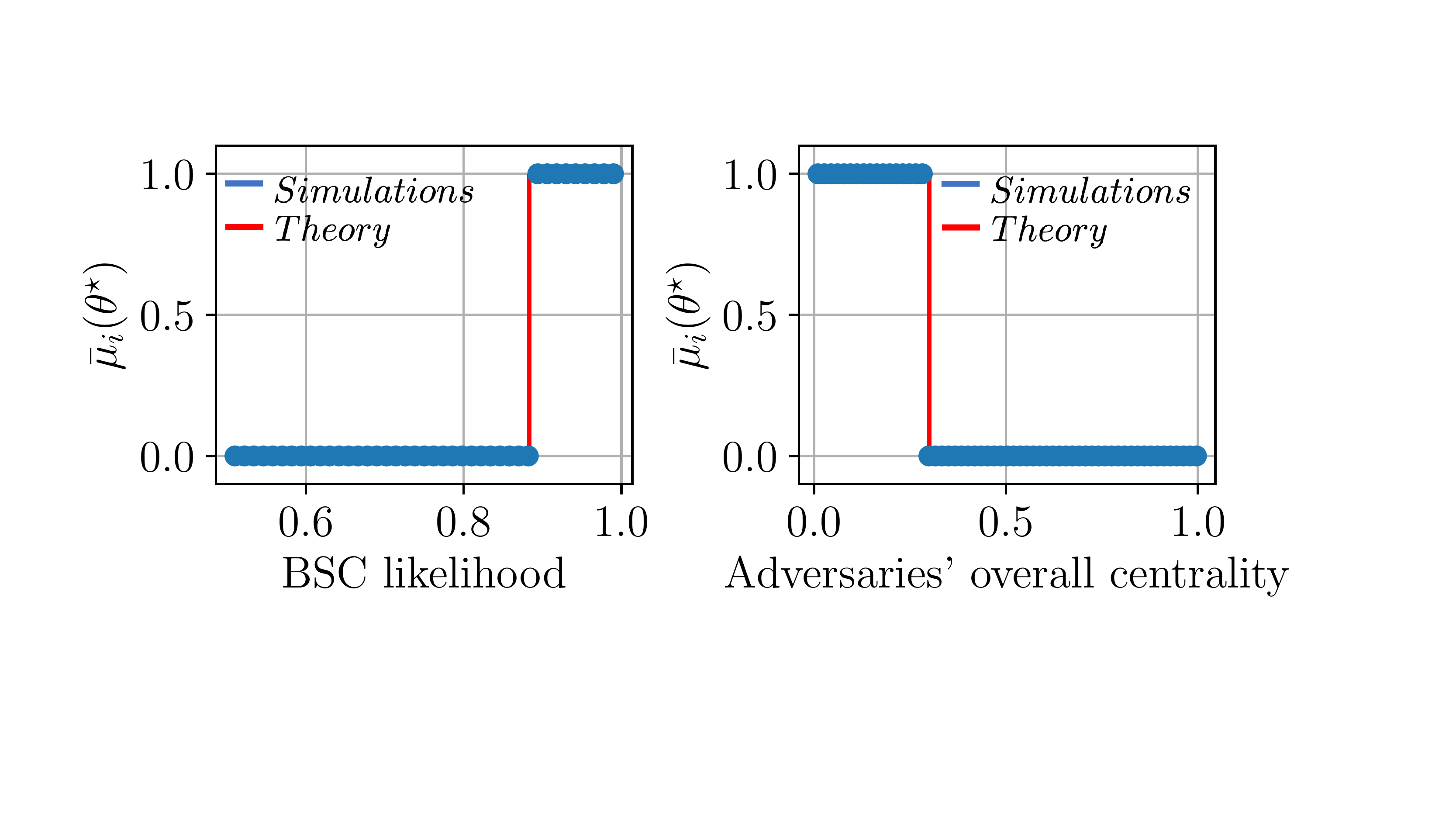}%
\vspace{-12mm}
\caption{Left: Agents' average steady-state beliefs for different BSC observation probabilities $p\in(0.5,1)$. The vertical red line shows the critical value of $p$ as predicted by Theory. Right: Agents' steady-state beliefs for different overall adversaries' centrality. The vertical red line shows the critical value of adversaries' centrality as predicted by Theory (by solving \eqref{lem_ratiosm} w.r.t. $p$ and $\sum_{k\in\mathcal{N}^m}u_{k}$, respectively).}
\label{obs_cen}
\end{figure}
\subsection{The role of separable observations}
In this subsection, we show the evolution of agents' beliefs when adversaries' observation models are {\em separable} ({\em see} Definition \ref{separation_definition}) and when they are not under the presented attack strategies. We note that for the BSC case, observations are always separable, since \eqref{condition_Th31}, \eqref{condition_Th32} are satisfied. In the left sub-figure of Fig. \ref{fig1}, the agents observe the state through a binary symmetric channel (BSC) with $p=0.9$. On the contrary, in the right sub-figure, the agents' observation models are given by $L(\zeta^1|\theta_1)=0.8$, $L(\zeta^2|\theta_1)=0.2$, $L(\zeta^1|\theta_1)=0.55$ and $L(\zeta^2|\theta_2)=0.45$ and thus observations are non-separable. We set $\epsilon=10^{-5}$ to illustrate the meaning of Theorem \ref{efficiciency_Th3}. We also note that, since the agents have the same observation models, all adversaries can be thought as one adversary and $u_{k}$ appearing in the expressions of Theorem \ref{region} is replaced with adversaries' overall centrality $\sum_{k\in\mathcal{N}^m}u_{k}$. We do this in order to avoid $\epsilon$ getting very small values. Also, we note that the value of $\epsilon=10^{-5}$ satisfies \eqref{e_condition} (in the computation of $x_1'$ and $x_2'$, $\sum_{k\in\mathcal{N}^m}u_{k}$ is used instead of $u_{k}$) for the specific setup. We observe that the network beliefs are driven to the wrong hypothesis for both cases when the true state is $\theta^{\star}=\theta_1$ and $\theta^{\star}=\theta_2$ in the left sub-figure for the attack strategies presented in Theorems \ref{region} and \ref{opt_attackth}. We also want to highlight that for the same setup the network was not misled under the attack strategy with unknown divergences in Fig. \ref{network_d} (Left sub-figures) as the value of $\epsilon=5\times10^{-3}$ in that setup was not sufficiently small.

In contrast to the case of separable observations, as we observe in the right sub-figure where \eqref{condition_Th31}, \eqref{condition_Th32} are not satisfied (i.e., observations are non-separable), the network is misled only for one state under the attack strategy with unknown divergences (Theorem \ref{opt_attackth}), while the attack strategy with known divergences (Theorem \ref{region}) drives the network to the wrong hypothesis for both $\theta^{\star}\in\Theta$, as expected. Finally, for comparison reasons we present here the random attack strategy where the distortion functions $\widehat{L}_k(\cdot\vert\theta_1),\widehat{L}_k(\cdot\vert\theta_2)$ are chosen randomly by the adversaries $k\in\mathcal{N}^m$. As we can see the impact of the random attack strategy is insufficient to mislead the network in any case, which highlights the need to appropriately design the fake likelihood functions so that the network is deceived.
\begin{figure}[!h]
\centering
\includegraphics[width=0.5\textwidth]{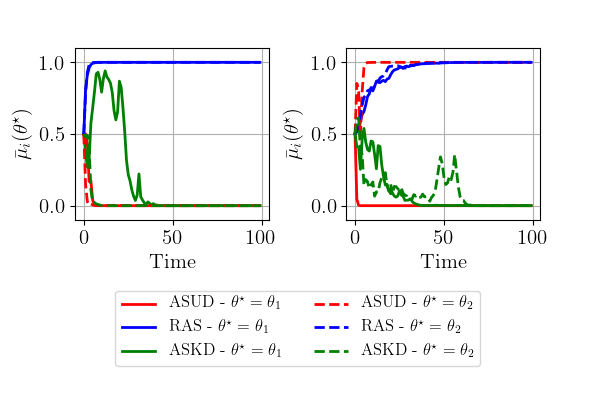}%
\vspace{-3mm}
\caption{Agents' average belief evolution. Left: Separable observations, Right: Non-separable observations. ASUD: %Optimal Attack Strategy
Attack Strategy with Unknown Divergences, RAS: Random Attack strategy, ASKD: Attack Strategy with Known Divergences.}
\label{fig1}
\end{figure}
\section{Conclusions}
In this paper, the impact of deceptive strategies on social learning was investigated. We characterized the evolution of agents' beliefs under inferential attacks and the adversaries' attack strategies were investigated. We showed that agnostic adversaries, which do not have knowledge about the true state, can always construct fake likelihood functions that provably drive the network to the wrong hypothesis, given that they have some knowledge about the normal sub-network properties. Then, we studied the case where the adversaries do not have any knowledge about the network properties. We formulated an optimization problem to derive the adversarial strategy for that scenario and provided performance guarantees for it. Finally, we illustrated the impact of adversarial strategies on the learning performance under different setups and highlighted the interplay among topology, informativeness of the signals and attack strategies that govern the learning performance of the network. Our results are expected to shed light on the study of more elaborate attack schemes as well as on the development of light-weight detection mechanisms based on agents' characteristics (i.e., network centrality and observation models) and provide useful insight to situations where networks compete with each other in a strategic fashion.

\appendices
\section{Proof of Theorem \ref{inconsistent_learning}}
\label{a1}
\noindent
We define the log-likelihood ratio for each agent $k\in\mathcal{N}$:
\begin{align}
\label{llr}
&\boldsymbol{\mathcal{L}}_{k,i}(\theta)\triangleq\log\frac{L_k'(\boldsymbol{\zeta}_{k,i}|\theta)}{L_k'(\boldsymbol{\zeta}_{k,i}|\theta^{\star})},\quad\theta\neq\theta^{\star},k\in\mathcal{N}^n
\end{align}
where $\theta^{\star}$ is the true state, $L_{\ell}'(\boldsymbol{\zeta}_{\ell,i}|\theta')=L_{\ell}(\boldsymbol{\zeta}_{\ell,i}|\theta')$ if $\ell\in\mathcal{N}^n$ and $L_{\ell}'(\boldsymbol{\zeta}_{\ell,i}|\theta')=\widehat{L}_{\ell}(\boldsymbol{\zeta}_{\ell,i}|\theta')$ if $\ell\in\mathcal{N}^m$ for all $\theta'\in\Theta$. 
We also define the log-belief ratio for each agent $k\in\mathcal{N}$ as
\begin{align}
%\label{log_ratio}
    \boldsymbol{\lambda}_{k,i}(\theta)\triangleq\log\frac{\boldsymbol{\mu}_{k,i}(\theta)}{\boldsymbol{\mu}_{k,i}(\theta^{\star})}, \quad\theta\neq\theta^{\star}.
\end{align}
Then, by utilizing \eqref{adapt}, \eqref{combine},  \eqref{adapt_malicious} and \eqref{llr}, the above equation yields
\begin{align}
\label{log_l2}
    &\boldsymbol{\lambda}_{k,i}(\theta)
    =\sum_{\ell\in\mathcal{N}_k}a_{\ell k}\boldsymbol{\mathcal{L}}_{\ell,i}(\theta)+
    \sum_{\ell\in\mathcal{N}_k}a_{\ell k}\log\frac{\boldsymbol{\mu}_{\ell,i-1}(\theta)}{\boldsymbol{\mu}_{\ell,i-1}(\theta^{\star})}.%\\
    % &=\sum_{\ell\in\mathcal{N}^n_k}a_{\ell k,i+1}\log\frac{L_{\ell}(\boldsymbol{\zeta}_{\ell,i+1}|\theta)}{L_{\ell}(\boldsymbol{\zeta}_{\ell,i+1}|\theta^{\star})}+
    % \sum_{\ell\in\mathcal{N}^m_k}a_{\ell k,i+1}\log\frac{\widehat{L}_{\ell}(\boldsymbol{\zeta}_{\ell,i+1}|\theta)}{\widehat{L}_{\ell}(\boldsymbol{\zeta}_{\ell,i+1}|\theta^{\star})}\nonumber\\
    % \label{log_l3}
    % &+\sum_{\ell\in\mathcal{N}^n_k}a_{\ell k,i+1}\log\frac{\boldsymbol{\mu}_{\ell,i}(\theta)}{\boldsymbol{\mu}_{\ell,i}(\theta^{\star})}+\sum_{\ell\in\mathcal{N}^m_k}a_{\ell k,i+1}\log\frac{\boldsymbol{\mu}_{\ell,i}(\theta)}{\boldsymbol{\mu}_{\ell,i}(\theta^{\star})}
\end{align}
Without loss of generality, we index first the malicious agents, followed by the normal ones. Eq. \eqref{log_l2} is written in matrix-vector notation as
\begin{align}
    \label{log_12mv}
        \boldsymbol{\lambda}_i(\theta)=A^{\mathsf{T}}\boldsymbol{\mathcal{L}}_i(\theta)+A^{\mathsf{T}}\boldsymbol{\lambda}_{i-1}(\theta)
\end{align}
where
{\begin{small}{
\begin{align}
&\boldsymbol{\mathcal{L}}_i(\theta)=\Bigg[\log\frac{\widehat{L}_{1,i}(\boldsymbol{\zeta}_{1,i}|\theta)}{\widehat{L}_{1,i}(\boldsymbol{\zeta}_{1,i}|\theta^{\star})},\ldots
,\log\frac{\widehat{L}_{|\mathcal{N}^m|,i}(\boldsymbol{\zeta}_{|\mathcal{N}^m|,i}|\theta)}{\widehat{L}_{|\mathcal{N}^m|,i}(\boldsymbol{\zeta}_{|\mathcal{N}^m|,i}|\theta^{\star})},\nonumber\\
&\log\frac{L_{|\mathcal{N}^m|+1,i}(\boldsymbol{\zeta}_{|\mathcal{N}^m|+1,i}|\theta)}{L_{|\mathcal{N}^m|+1,i}(\boldsymbol{\zeta}_{|\mathcal{N}^m|+1,i}|\theta^{\star})},\ldots,\log\frac{L_{|\mathcal{N}|,i}(\boldsymbol{\zeta}_{|\mathcal{N}|,i}|\theta)}{L_{|\mathcal{N}|,i}(\boldsymbol{\zeta}_{|\mathcal{N}|,i}|\theta^{\star})}\Bigg]^{\mathsf{T}}
\end{align}
}\end{small}}
and $\boldsymbol{\lambda}_i(\theta)=[\boldsymbol{\lambda}_{1,i}(\theta),\ldots,\boldsymbol{\lambda}_{|\mathcal{N}|,i}(\theta)]^{\mathsf{T}}$.
Iterating \eqref{log_12mv} yields
\begin{align}
    \boldsymbol{\lambda}_i(\theta)=%A^T_{i+1}\boldsymbol{\mathcal{L}}_{i+1}(\theta)+
    \sum^i_{t=1}(A^{\mathsf{T}})^{t}\boldsymbol{\mathcal{L}}_{i-t+1}(\theta)+(A^{\mathsf{T}})^{i}\lambda_0(\theta)%\nonumber
\end{align}
where $\lambda_0$ is assumed to be a deterministic initial vector. Dividing by $i$ and taking the limit as $i\rightarrow\infty$ yields
\begin{align}
    \label{logs_rec}
    &\lim_{i\rightarrow\infty}\frac{1}{i}\boldsymbol{\lambda}_i(\theta)=%\nonumber\\
    \lim_{i\rightarrow\infty}\frac{1}{i}\sum^i_{t=1}(A^{\mathsf{T}})^{t}\boldsymbol{\mathcal{L}}_{i-t+1}(\theta)\nonumber\\
    &+\lim_{i\rightarrow\infty}\frac{1}{i}(A^{\mathsf{T}})^{i}\lambda_0(\theta).
\end{align}
Since $\mu_{k,0}(\theta)>0$ %, for all $\theta\in\Theta$, $k\in\mathcal{N}$ 
(Assumption \ref{positive_beliefs}), the second term in RHS of \eqref{logs_rec} goes to $0$ as $i\to\infty$. The first term of \eqref{logs_rec} can be rewritten as
\begin{align}
\label{rec_lim}
    &\lim_{i\rightarrow\infty}\frac{1}{i}\sum^i_{t=1}(A^{\mathsf{T}})^{t}\boldsymbol{\mathcal{L}}_{i-t+1}(\theta)\nonumber\\
    &=\lim_{i\rightarrow\infty}\frac{1}{i}\sum^i_{t=1}(A^{t}-u\mathds{1}^{\mathsf{T}})^{\mathsf{T}}\boldsymbol{\mathcal{L}}_{i-t+1}(\theta)\nonumber\\
    &+\lim_{i\rightarrow\infty}\frac{1}{i}\sum^i_{t=1}\mathds{1}u^{\mathsf{T}}\boldsymbol{\mathcal{L}}_{i-t+1}(\theta)
\end{align}
We note that $\boldsymbol{\mathcal{L}}_{\ell,t}(\theta)$ has finite expectation for every $\ell\in\mathcal{N}^n$ from Assumption \ref{obs_independence} and $\boldsymbol{\mathcal{L}}_{\ell,t}(\theta)$ is a bounded random variable for every $\ell\in\mathcal{N}^m$ from Assumption \ref{finiteness} and as a result has finite expectation. Thus, $\boldsymbol{\mathcal{L}}_t(\theta)$ are i.i.d. random vectors with finite expectation and from the strong law of large numbers (SLLN) the second term in RHS of \eqref{rec_lim} yields
\begin{align}
    \hspace{-1mm}\frac{1}{i}\sum^i_{t=1}\mathds{1}u^{\mathsf{T}}\hspace{-0.9mm}\boldsymbol{\mathcal{L}}_{i-t+1}(\theta)\hspace{-1mm}=\hspace{-1mm}
    \frac{1}{i}\sum^{i}_{t=1}\hspace{-0.9mm}\mathds{1}u^{\mathsf{T}}\boldsymbol{\mathcal{L}}_{t}(\theta)\hspace{-0.9mm}
    \overset{\text{a.s.}}\rightarrow\hspace{-0.9mm}\mathds{1}u^{\mathsf{T}}\mathbb{E}\{\boldsymbol{\mathcal{L}}_{t}(\theta)\}.
\end{align}
Now we will show that the first term in the RHS of \eqref{rec_lim} goes to $0$ (a similar argument is found in Lemma 8 from \cite{Virginia_2020b}). 
First, from Assumption~\ref{strongly_connected}, we know that 
\begin{equation}
\lim_{t\to\infty}A^t=u\mathds{1}^{\sf T}
\end{equation}
which in turn implies that, for $\varepsilon>0$, there exists a time index $t_0$ such that for all $t>t_0$:
\begin{equation}
    \left|[A^t]_{\ell k}-u_\ell \right|<\varepsilon\label{eq:convA}
\end{equation}
where $[A^t]_{\ell k}$ denotes the element $\ell k$ of matrix $A^t$. We can write:
\begin{align}
    &\frac{1}{i}\sum^i_{t=1}(A^{t}-u\mathds{1}^{\mathsf{T}})^{\mathsf{T}}\boldsymbol{\mathcal{L}}_{i-t+1}(\theta)\nonumber\\
    &=\frac{1}{i}\sum^i_{t=t_0+1}(A^{t}-u\mathds{1}^{\mathsf{T}})^{\mathsf{T}}\boldsymbol{\mathcal{L}}_{i-t+1}(\theta)\nonumber\\&+\frac{1}{i}\sum^{t_0}_{t=1}(A^{t}-u\mathds{1}^{\mathsf{T}})^{\mathsf{T}}\boldsymbol{\mathcal{L}}_{i-t+1}(\theta).\label{eq:conveq}
\end{align}
In view of \eqref{eq:convA}, we can write the absolute value of each component $k$ of the first term on the RHS of \eqref{eq:conveq} as
\begin{align}
\label{first_term}
    &\frac{1}{i}\left|\sum^i_{t=t_0+1}\sum_{\ell\in \mathcal{N}}([A^{t}]_{\ell k}-u_\ell)\boldsymbol{\mathcal{L}}_{\ell, i-t+1}(\theta)\right|\nonumber\\
    &\leq\hspace{-0.9mm}\varepsilon\hspace{-0.9mm}\sum_{\ell\in \mathcal{N}}\hspace{-0.9mm}\frac{1}{i}\hspace{-0.9mm}\sum^i_{t=t_0+1}|\boldsymbol{\mathcal{L}}_{\ell,i-t+1}(\theta)|\hspace{-0.9mm}=\hspace{-0.9mm}\varepsilon \sum_{\ell\in \mathcal{N}}\frac{1}{i}\sum^{i-t_0}_{t=1}|\boldsymbol{\mathcal{L}}_{\ell, t}(\theta)|.
\end{align}
As stated earlier, from Assumptions \ref{obs_independence} and \ref{finiteness}, $\boldsymbol{\mathcal{L}}_{\ell,t}(\theta)$ has finite expectation for all $\ell\in\mathcal{N}$ and as a result, 
$|\boldsymbol{\mathcal{L}}_{\ell,t}(\theta)|$ has finite expectation \cite{Billingsley}. In view of the i.i.d. property of variable $|\boldsymbol{\mathcal{L}}_{\ell,t}(\theta)|$, from the SLLN we have that
\begin{align}
&\frac{1}{i}\sum_{t=1}^{i-t_0}|\boldsymbol{\mathcal{L}}_{\ell,t}(\theta)|\nonumber\\
&\hspace{-1mm}=\frac{i-t_0}{i}\frac{1}{i-t_0}\sum_{t=1}^{i-t_0}|\boldsymbol{\mathcal{L}}_{\ell,t}(\theta)|\stackrel{\text{a.s.}}{\longrightarrow}\mathbb{E}\{|\boldsymbol{\mathcal{L}}_{\ell,t}(\theta)|\}<+\infty.\label{eq:integrconv}
\end{align}
From \eqref{first_term} and \eqref{eq:integrconv} we have
\begin{align}
\label{lim_first}
    &\limsup_{i\to\infty}\frac{1}{i}\left|\sum^i_{t=t_0+1}\sum_{\ell\in \mathcal{N}}([A^{t}]_{\ell k}-u_\ell)\boldsymbol{\mathcal{L}}_{\ell, i-t+1}(\theta)\right|\nonumber\\
    &\overset{\text{a.s.}}\leq \varepsilon \sum_{\ell\in \mathcal{N}}\mathbb{E}\{|\boldsymbol{\mathcal{L}}_{\ell, t}(\theta)|\}
\end{align}
where $\sum_{\ell\in \mathcal{N}}\mathbb{E}\{|\boldsymbol{\mathcal{L}}_{\ell, t}(\theta)|\}$ is independent of $\epsilon$. % and \eqref{lim_first} holds for any $\epsilon>0$. 
Taking the limit as $\epsilon\to0$, we conclude that the limit superior vanishes and therefore, the first term on the RHS of \eqref{eq:conveq} vanishes. 
Next, taking the absolute value of the second term on the RHS of \eqref{eq:conveq} yields for each component $k$:
\begin{align}
    &\frac{1}{i}\left| \sum^{t_0}_{t=1}\sum_{\ell\in \mathcal{N}}([A^{t}]_{\ell k}-u_\ell)\boldsymbol{\mathcal{L}}_{\ell,i-t+1}(\theta)\right|\nonumber\\
    &\leq2\sum_{\ell\in \mathcal{N}}\frac{1}{i} \sum^{t_0}_{t=1}\left|\boldsymbol{\mathcal{L}}_{\ell,i-t+1}(\theta)\right|\label{eq:integrconv1}
\end{align}
where the last step follows from the fact that $A$ is left stochastic and the elements of the Perron eigenvector $u$ are strictly smaller than $1$. We can decompose the expression contained in the RHS of \eqref{eq:integrconv1} as:
\begin{align}
    &\frac{1}{i}\sum^{t_0}_{t=1}\left|\boldsymbol{\mathcal{L}}_{\ell,i-t+1}(\theta)\right|=\frac{1}{i}\sum^i_{t=1}\left|\boldsymbol{\mathcal{L}}_{\ell,i-t+1}(\theta)\right|\label{eq:sllncont}\nonumber\\
    &-\frac{1}{i}\sum^i_{t=t_0+1}\left|\boldsymbol{\mathcal{L}}_{\ell,i-t+1}(\theta)\right|%\nonumber\\
    =\frac{1}{i}\sum^{i}_{t=1}\left|\boldsymbol{\mathcal{L}}_{\ell,t}(\theta)\right|-\frac{1}{i}\sum^{i-t_0}_{t=1}\left|\boldsymbol{\mathcal{L}}_{\ell,t}(\theta)\right|.
\end{align}
From the SLLN both terms on the RHS of \eqref{eq:sllncont} go almost surely to $\mathbb{E}\{|\boldsymbol{\mathcal{L}}_{\ell, t}(\theta)|\}$, which implies by the continuous mapping theorem (\cite{Shao_2003} - Theorem 1.10) that
\begin{equation}
\label{lim2}
     \frac{1}{i}\sum^{t_0}_{t=1}\left|\boldsymbol{\mathcal{L}}_{\ell,i-t+1}(\theta)\right|\stackrel{\text{a.s.}}{\longrightarrow}0.
\end{equation}
Hence, \eqref{logs_rec} yields
% \begin{small}
\begin{align}
\label{lim_logasymmetric}
    &\lim_{i\rightarrow\infty}\frac{1}{i}\boldsymbol{\lambda}_i(\theta)\stackrel{a.s.}{\longrightarrow}\mathds{1}u^{\mathsf{T}}\mathbb{E}\{\boldsymbol{\mathcal{L}}_t\}=%\nonumber\\
    %&
    \sum^N_{\ell=1}u_{\ell}\mathbb{E}\{\boldsymbol{\mathcal{L}}_{\ell,t}(\theta)\}\mathds{1}\nonumber\\
    &=\Bigg(\sum_{\ell\in\mathcal{N}^n}u_{\ell}\mathbb{E}\Bigg\{\log\frac{L_{\ell}(\boldsymbol{\zeta}_{\ell}|\theta)}{L_{\ell}(\boldsymbol{\zeta}_{\ell}|\theta^{\star})}\Bigg\}\nonumber\\
    &+\sum_{\ell\in\mathcal{N}^m}u_{\ell}\mathbb{E}\Bigg\{\log\frac{\widehat{L}_{\ell}(\boldsymbol{\zeta}_{\ell}|\theta)}{\widehat{L}_{\ell}(\boldsymbol{\zeta}_{\ell}|\theta^{\star})}\Bigg\}\Bigg)\mathds{1}.
\end{align}
Observing \eqref{lim_logasymmetric} we conclude the following.
{\begin{small}{
\begin{align}
\label{req1_asm}
    &\sum_{\ell\in\mathcal{N}^n}u_{\ell}\mathbb{E}\Big\{\log\frac{{L_{\ell}(\boldsymbol{\zeta}_{\ell}|\theta^{\star})}}{{L_{\ell}(\boldsymbol{\zeta}_{\ell}|\theta)}}\Big\}>
    \sum_{\ell\in\mathcal{N}^m}u_{\ell}\mathbb{E}\Big\{\log\frac{\widehat{L}_{\ell}(\boldsymbol{\zeta}_{\ell}|\theta)}{\widehat{L}_{\ell}(\boldsymbol{\zeta}_{\ell}|\theta^{\star})}\Big\}\nonumber\\
    &\Rightarrow\boldsymbol{\lambda}_{\ell,i}(\theta)\stackrel{a.s.}{\longrightarrow}-\infty\Rightarrow\boldsymbol{\mu}_{\ell,i}(\theta)\stackrel{a.s.}{\longrightarrow} 0\Rightarrow\boldsymbol{\mu}_{\ell,i}(\theta^{\star})\stackrel{a.s.}{\longrightarrow} 1
    ,\forall \ell\in\mathcal{N}
    \end{align}
     }\end{small}}%
    and
    {\begin{small}{
    \begin{align}
    \label{req2_asm}
    &\sum_{\ell\in\mathcal{N}^n}u_{\ell}\mathbb{E}\Bigg\{\log\frac{{L_{\ell}(\boldsymbol{\zeta}_{\ell}|\theta^{\star})}}{{L_{\ell}(\boldsymbol{\zeta}_{\ell}|\theta)}}\Bigg\}<
    \sum_{\ell\in\mathcal{N}^m}u_{\ell}\mathbb{E}\Bigg\{\log\frac{\widehat{L}_{\ell}(\boldsymbol{\zeta}_{\ell}|\theta)}{\widehat{L}_{\ell}(\boldsymbol{\zeta}_{\ell}|\theta^{\star})}\Bigg\}\nonumber\\
    &\Rightarrow\boldsymbol{\lambda}_i(\theta)\stackrel{a.s.}{\longrightarrow}+\infty\Rightarrow\boldsymbol{\mu}_{\ell,i}(\theta^{\star})\stackrel{a.s.}{\longrightarrow}0, \forall \ell\in\mathcal{N}.
\end{align}
}\end{small}}
This concludes the proof.\qedsymb
\section{Proof of Lemma \ref{Lem_uninformative}}
\label{a2}
Let us assume that there is only one adversary in the network denoted by $k$ (i.e., $\mathcal{N}^m=\{k\}$). Then condition \eqref{lem_ratiosm} is equivalent to \eqref{lem_ratiosm30}, which yields:
\begin{align}
&S_1<u_{k}\sum_{\zeta_{k}\in\mathcal{Z}_{k}}L_{k}(\zeta_{k}|\theta_1)\log\frac{\widehat{L}_{k}(\zeta_{k}|\theta_2)}{\widehat{L}_{k}(\zeta_{k}|\theta_1)}=R_{k,1}\\
&S_2<u_{k}\sum_{\zeta_{k}\in\mathcal{Z}_{k}}L_{k}(\zeta_{k}|\theta_2)\log\frac{\widehat{L}_{k}(\zeta_{k}|\theta_1)}{\widehat{L}_{k}(\zeta_{k}|\theta_2)}=R_{k,2}.
\end{align}
We show the result by contradiction. Let us assume that the corrupted PMFs $\widehat{L}_{k}(\cdot\vert\theta_1),\widehat{L}_{k}(\cdot\vert\theta_2)$ satisfy the inequalities above. Then, it follows that both $R_{k,1}>0$ and $R_{k,2}>0$ since $S_1$ and $S_2$ are non-negative as they are positive weighted sums of KL divergences. 

However, agent $k$'s observations are uninformative, and as a result we have:
\begin{align}
\label{contradiction_eq}
    &R_{k,2}=u_{k}\sum_{\zeta_{k}\in\mathcal{Z}_{k}}L_{k}(\zeta_{k}|\theta_2)\log\frac{\widehat{L}_{k}(\zeta_{k}|\theta_1)}{\widehat{L}_{k}(\zeta_{k}|\theta_2)}\nonumber\\
    &\overset{(a)}=-u_{k}\sum_{\zeta_{k}\in\mathcal{Z}_{k}}L_{k}(\zeta_{k}|\theta_1)\log\frac{\widehat{L}_{k}(\zeta_{k}|\theta_2)}{\widehat{L}_{k}(\zeta_{k}|\theta_1)}=-R_{k,1}
\end{align}
where $(a)$ is true due to the fact that if agent $k$ has uninformative PMFs, then  $L_k(\zeta_k\vert\theta_1)=L_k(\zeta_k\vert\theta_2)$ for all $\zeta_k\in\mathcal{Z}_k$. Eq.  \eqref{contradiction_eq} leads to a contradiction, since it implies that $R_{k,1}$ and $R_{k,2}$ cannot be both positive.

The result extends in a straightforward way to the case where all adversaries have uninformative PMFs, by using the above reasoning and by utilizing \eqref{inequality_re} instead of \eqref{lem_ratiosm30}.\qedsymb
\section{Proof of Theorem \ref{region}}
\label{a23}
We prove first an auxiliary lemma, which establishes that for an adversary $k$ with informative PMFs there always exist realizations $\zeta^1,\zeta^2\in\mathcal{Z}_k$ such that \eqref{det_positive} holds.
\begin{Lem}
\label{non_zerodet}{\bf Existence of appropriate realizations of $\boldsymbol{\zeta}\in\mathcal{Z}_k$}. 
For an agent $k\in\mathcal{N}$ with informative PMFs, there always exist $\zeta^1\neq\zeta^2$, $\zeta^1,\zeta^2\in\mathcal{Z}_k$ such that $L_{k}(\zeta^1\vert\theta_1)L_{k}(\zeta^2\vert\theta_2)\neq L_{k}(\zeta^1\vert\theta_2)L_{k}(\zeta^2\vert\theta_1)$.
\end{Lem}
\begin{proof}
For an agent $k\in\mathcal{N}$ with informative PMFs there exists at least one observation $\zeta^1\in\mathcal{Z}_{k}$ such that:
\begin{align}
    \label{cs1}
    L_{k}(\zeta^1\vert\theta_1)>L_{k}(\zeta^1\vert\theta_2)
\end{align}
or
\begin{align}
\label{cs2}
    L_{k}(\zeta^1\vert\theta_1)<L_{k}(\zeta^1\vert\theta_2).
\end{align}
Otherwise the PMFs would be uninformative. We focus on the first case (i.e., \eqref{cs1}) and the second case (i.e., \eqref{cs2}) follows accordingly. Let us assume that there exists no $\zeta\in\mathcal{Z}_{k}$ such that $L_{k}(\zeta\vert\theta_1)<L_{k}(\zeta\vert\theta_2)$. Then,
\begin{align}
    \sum_{\zeta}L_{k}(\zeta\vert\theta_1)=1>\sum_{\zeta}L_{k}(\zeta\vert\theta_2)
\end{align}
which cannot hold, since $L_{k}(\cdot\vert\theta)$, $\theta\in\Theta$ is a PMF. Thus, there always exists at least one observation $\zeta^2\in\mathcal{Z}_{k}$ such that
\begin{align}
    L_{k}(\zeta^2\vert\theta_1)<L_{k}(\zeta^2\vert\theta_2).
\end{align}
Then, it follows that
\begin{align}
    L_{k}(\zeta^1\vert\theta_1)L_{k}(\zeta^2\vert\theta_2)> L_{k}(\zeta^1\vert\theta_2)L_{k}(\zeta^2\vert\theta_1).
\end{align}
In the second case (i.e., \eqref{cs2}) we would arrive at the same inequality as above with opposite direction (i.e., $<$). 
Thus, we conclude that for an agent $k$ with informative PMFs there always exist $\zeta^1,\zeta^2\in\mathcal{Z}_{k}$ such that
\begin{align}
    L_{k}(\zeta^1\vert\theta_1)L_{k}(\zeta^2\vert\theta_2)\neq L_{k}(\zeta^1\vert\theta_2)L_{k}(\zeta^2\vert\theta_1).
\end{align}
% \qedsymb
\end{proof}
Now, we proceed with the proof of Theorem \ref{region}. For $|\mathcal{N}^m|=1$ \eqref{lem_ratiosm} is equivalent to \eqref{lem_ratiosm30}. Then, for the construction \eqref{construction10}, \eqref{lem_ratiosm30} yields
% {\begin{small}{
\begin{align}
\label{ineq1b}
    &\log\frac{\alpha_{k}-p_{k,1}}{p_{k,2}}>\frac{S_1}{u
    _kL_{k}(\zeta^2_{k}\vert\theta_1)}-\frac{L_{k}(\zeta^1_{k}\vert\theta_1)}{L_{k}(\zeta^2_{k}\vert\theta_1)}\log\frac{p_{k,1}}{\alpha_{k}-p_{k,2}}\\
    \label{ineq2b}
    &\log\frac{\alpha_{k}-p_{k,1}}{p_{k,2}}<-\frac{S_2}{u_{k}L_{k}(\zeta^2_{k}\vert\theta_2)}-\frac{L_{k}(\zeta^1_{k}\vert\theta_2)}{L_{k}(\zeta^2_{k}\vert\theta_2)}\log\frac{p_{k,1}}{\alpha_{k}-p_{k,2}}.
\end{align}
% }\end{small}}%
where $\alpha_k=1-(|\mathcal{Z}_k|-2)\epsilon$. Note that the rest of the terms in RHS of \eqref{lem_ratiosm30} (i.e., $L_k(\zeta_k|\theta^{\star})\log\frac{\widehat{L}_k(\zeta_k\vert\theta^{\star})}{\widehat{L}_k(\zeta_k\vert\theta)}, \zeta_k\neq\zeta^1_k,\zeta^2_k$) vanish due to choice $\widehat{L}_k(\zeta_k\vert\theta^{\star})=\widehat{L}_k(\zeta_k\vert\theta)=\epsilon$, for every $\theta^{\star},\theta\in\Theta,\theta^{\star}\neq\theta$ and for every $\zeta_k\neq\zeta^1_k,\zeta^2_k$. We observe that the above system of inequalities is non-linear w.r.t. $p_{k,1},p_{k,2}$. However, the system is linear w.r.t. to the log-likelihood ratios $\log\frac{p_{k,1}}{\alpha_{k}-p_{k,2}},\log\frac{\alpha_{k}-p_{k,1}}{p_{k,2}}$. Motivated by this observation, we introduce:
\begin{align}
\label{x_1}
&x_1\triangleq\log\frac{p_{k,1}}{\alpha_k-p_{k,2}}\\
\label{x_2}
&x_2\triangleq\log\frac{\alpha_k-p_{k,1}}{p_{k,2}}.
\end{align}
Next, instead of solving the system \eqref{ineq1b}, \eqref{ineq2b} directly w.r.t. $p_{k,1},p_{k,2}$, we  
choose to solve the following system of inequalities w.r.t. $x_1,x_2$.
\begin{align}
\label{sys11}
    &x_2>\frac{S_1}{u
    _kL_{k}(\zeta^2_{k}\vert\theta_1)}-\frac{L_{k}(\zeta^1_{k}\vert\theta_1)}{L_{k}(\zeta^2_{k}\vert\theta_1)}x_1\\
    \label{sys22}
    &x_2<-\frac{S_2}{u_{k}L_{k}(\zeta^2_{k}\vert\theta_2)}-\frac{L_{k}(\zeta^1_{k}\vert\theta_2)}{L_{k}(\zeta^2_{k}\vert\theta_2)}x_1
\end{align}
with $x_1,x_2\in\mathbb{R}$ (we do not restrict $x_1,x_2$ to satisfy \eqref{x_1}, \eqref{x_2}, respectively). After solving the above system, we will show that for every $x_1,x_2\in\mathbb{R}$ satisfying \eqref{sys11} and \eqref{sys22}, there exist appropriate probability values $p_{k,1},p_{k,2}$ that satisfy \eqref{x_1} and \eqref{x_2}. Let us define the region
\begin{align}
\widetilde{\mathcal{R}}_k\triangleq\{ (x_1,x_2)\in\mathbb{R}: \eqref{sys11}, \eqref{sys22} \text{ hold}\}.
\end{align}
$\widetilde{\mathcal{R}}_k$ is defined by the following linear relations:
\begin{align}
\label{line1}
    &x_2=\frac{S_1-u_kL_k(\zeta^1_k\vert\theta_1)x_1}{u
    _kL_k(\zeta^2_k\vert\theta_1)}\triangleq r_1(x_1)\\
\label{line2}
    &x_2=-\frac{S_2+u_kL_k(\zeta^1_k\vert\theta_2)x_1}{u_kL_k(\zeta^2_k\vert\theta_2)}\triangleq r_2(x_1)
\end{align}
The intersection point is
% \begin{small}
\begin{align}
    &(x_1',x_2')=\Big(\frac{n_2}{u_kd_k},-\frac{n_1}{u_kd_k}\Big)
\end{align}
% \end{small}%
where $n_j=L_k(\zeta^j_k\vert\theta_2)S_1+L_k(\zeta^j_k\vert\theta_1)S_2$, $j=1,2$ and $d_k=L_k(\zeta^2_k\vert\theta_2)L_k(\zeta^1_k\vert\theta_1)-L_k(\zeta^2_k\vert\theta_1)L_k(\zeta^1_k\vert\theta_2)$. Note that $x_1',x_2'<\infty$ since $d_k\neq0$ for an agent with informative PMFs from Lemma \ref{non_zerodet}, $u_k\neq0$ since it is an entry of the Perron eigenvector and $S_1,S_2<\infty$ from Assumption \ref{obs_independence}. Both slopes are negative and $r_1(x_1),r_2(x_1)$ intersect the $x_2$ axis at points
\begin{align}
&r_1(x_1=0)=\frac{S_1}{u_kL_k(\widehat{\zeta}_k\vert\theta_1)}\geq0\\
&r_2(x_1=0)=-\frac{S_2}{u_kL_k(\widehat{\zeta}_k\vert\theta_2)}\leq0
\end{align}
respectively. % (as $S_1,S_2\geq0$). 
Thus, the region $\widetilde{\mathcal{R}}_k$ is given by:
\begin{align}
% {\begin{small}{
    \widetilde{\mathcal{R}}_k\hspace{-1mm}=\hspace{-1mm}\begin{cases}\{(x_1,x_2):x_1<x_1',r_1(x_1)<
    x_2<r_2(x_1),\text{ if } d_k<0,\\
    \{(x_1,x_2):x_1>x_1',r_1(x_1)<
    x_2<r_2(x_1),\text{ if } d_k>0.
    \end{cases}
% }\end{small}}
\end{align}
Next, we show that for any $x_1,x_2\in\widetilde{\mathcal{R}}_k$ the following is true
\begin{align}
\label{e1_pr}
    &0<p_{k,1}<1,\\
    \label{e2_pr}
    &0<p_{k,2}<1.
\end{align}
Solving \eqref{x_1}, \eqref{x_2} w.r.t. $p_{k,1},p_{k,2}$ we get
\begin{align}
\label{e1}
    &p_{k,1}=\frac{e^{x_1}\alpha_k(e^{x_2}-1)}{e^{x_2}-e^{x_1}}\\
\label{e2}
    &p_{k,2}=\frac{\alpha_k(1-e^{x_1})}{e^{x_2}-e^{x_1}}.
\end{align}
Let $d_k<0$
. Then, we have for every $x_1,x_2\in\widetilde{\mathcal{R}}_k$
\begin{align}
    &x_1<x_1'<0\iff e^{x_1}<1\\
    &x_2>x_2'>0\iff e^{x_2}>1.
\end{align}
The second inequality holds because $x_2'>0$ and both slopes of $r_1(x_1)$ and $r_2(x_1)$ are negative. Then, utilizing \eqref{e1}, \eqref{e2} we can verify that $p_{k,1},p_{k,2}>0$. Moreover, we have
\begin{align}
    &p_{k,1}=\frac{e^{x_1}\alpha_k(e^{x_2}-1)}{e^{x_2}-e^{x_1}}<\frac{e^{x_1}(e^{x_2}-1)}{e^{x_2}-e^{x_1}}<1\\
    &p_{k,2}=\frac{\alpha_k(1-e^{x_1})}{e^{x_2}-e^{x_1}}<\frac{1-e^{x_1}}{e^{x_2}-e^{x_1}}<1
\end{align}
Thus, \eqref{e1_pr}, \eqref{e2_pr} hold.

On the other hand, if $d_k>0$ we have for every $x_1,x_2\in\widetilde{\mathcal{R}}_k$
\begin{align}
&x_1>x_1'>0\iff e^{x_1}>1\\
&x_2<x_2'<0\iff e^{x_2}<1.
\end{align}
Then, working in a similar fashion as previously, from \eqref{e1}, \eqref{e2} we can verify that \eqref{e1_pr}, \eqref{e2_pr} hold.

Thus, for every $x_1,x_2\in\widetilde{\mathcal{R}}_k$, there exist $p_{k,1},p_{k,2}$ satisfying \eqref{x_1}, \eqref{x_2} such that $0<p_{k,1},p_{k,2}<1$
. However, the admissible values of $p_{k,1},p_{k,2}$ are restricted due to Assumption \ref{finiteness} and the fact that $\widehat{L}_k(\zeta_k\vert\theta)$ must sum up to one over $\zeta_k$ for all $\theta\in\Theta$, which imply that
\begin{align}
\epsilon\leq p_{k,1},p_{k,2}\leq1-(|\mathcal{Z}_k|-1)\epsilon.
\end{align}
Then, for $x_1,x_2$ given by \eqref{x_1}, \eqref{x_2} we have that $x_1,x_2\in[x^-,x^+]$, where
\begin{align}
&x^-\triangleq\log\frac{\epsilon}{1-(|\mathcal{Z}_k|-1)\epsilon}\\ &x^+\triangleq\log\frac{1-(|\mathcal{Z}_k|-1)\epsilon}{\epsilon}=-x^-.
\end{align}
Thus, the set of admissible values of $x_1,x_2$ given by \eqref{x_1}, \eqref{x_2} that satisfy \eqref{sys11} and \eqref{sys22} are given by
\begin{align}
    \mathcal{R}_k\triangleq\widetilde{\mathcal{R}}_k\cap\{(x_1,x_2):x_1,x_2\in[x^-,x^+]\}.
\end{align}
We call $\mathcal{R}_k$ {\em distortion region} for agent $k$. A geometrical representation of the distortion region is shown in Figs. \ref{diag1}, \ref{diag2}.
\begin{figure}[!h]
\centering
\includegraphics[width=0.5\textwidth]{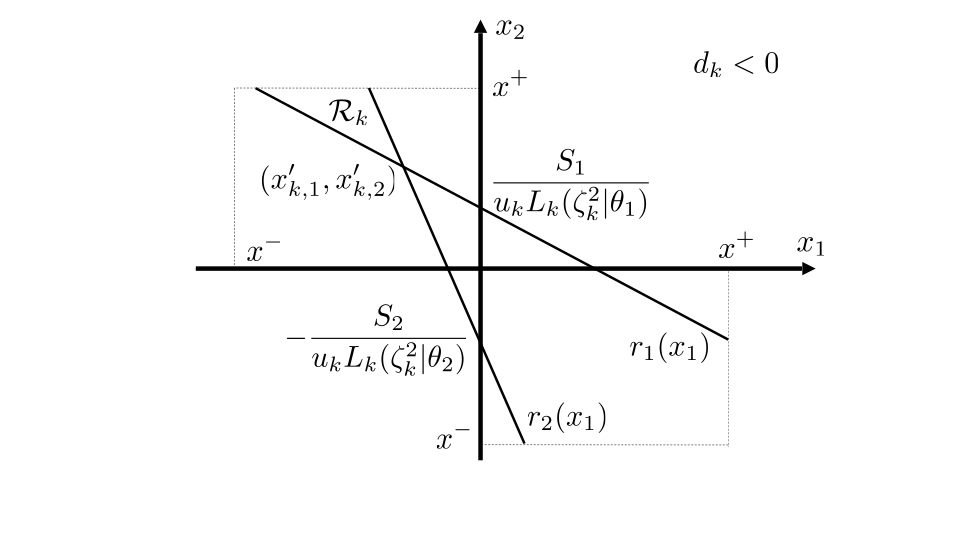}
\caption{Geometrical representation of the distortion region $\mathcal{R}_{k}$ for $d_{k}<0$.}
\label{diag1}
\end{figure}
\begin{figure}[!h]
% \label{diag2}
\centering
\includegraphics[width=0.5\textwidth]{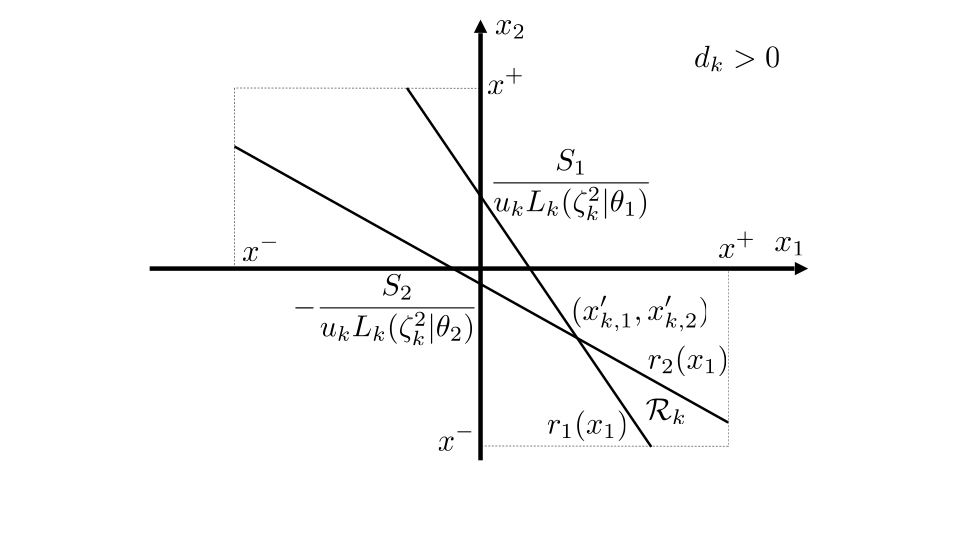}
\caption{Geometrical representation of the distortion region $\mathcal{R}_{k}$ for $d_{k}>0$.}
\label{diag2}
\end{figure}

Now, we will show that there always exists $\epsilon>0$ such that $\mathcal{R}_k$ is non-empty. In order for $\mathcal{R}_k$ to be non-empty, the point $(x_1',x_2')$ must be inside the space of admissible values for $x_1,x_2$. Thus, we have
\begin{align}
\label{c1}
&x_1'<x^+\quad\text{and}\quad x_2'>x^-,\quad\text{if}\quad d_k>0
\\
\label{c2}
&x_1'>x^-\quad\text{and}\quad x_2'<x^+,\quad\text{if}\quad d_k<0.
\end{align}
Both \eqref{c1} and \eqref{c2} are satisfied if the following holds.
\begin{align}
\label{eq_t}
    \epsilon<\min\left\{\frac{1}{e^{|x_1'|}+|\mathcal{Z}_k|-1},\frac{1}{e^{|x_2'|}+|\mathcal{Z}_k|-1}\right\}.
\end{align}
$S_1,S_2$ are finite from Assumption \ref{obs_independence} and thus, $|x_1'|,|x_2'|$ are also finite. Hence, the RHS of \eqref{eq_t} is strictly positive and as a result, we conclude that there always exists $\epsilon>0$ such that $\mathcal{R}_k$ is non-empty.

Finally, we will provide a way to construct $p_{k,1},p_{k,2}$ 
that satisfy \eqref{ineq1b}, \eqref{ineq2b}. In doing so, we will provide a way to select appropriate $x_1,x_2\in\mathcal{R}_k$ and then, we can retrieve some appropriate $p_{k,1},p_{k,2}$, which we know by the previous analysis that always exist. 
Adversary $k$ selects a $x_1$ such that $x^+>x_1>x_1'$ if $d_k>0$ and $x^-<x_1<x_1'$ if $d_k<0$. Then, it can select any $x_2\in \mathcal{R}_k$. The values $(x_1,x_2)$ that lie in $\mathcal{R}_k$ can be represented in a linear relation as
% \begin{small}
\begin{align}
\label{x2p}
    x_2=\beta (x_1-x_1')+x_2'
\end{align}
% \end{small}%
with
% \begin{small}
\begin{align}
    \min_j\Big\{-\frac{L_k(\zeta^1\vert\theta_j)}{L_k(\zeta^2\vert\theta_j)}\Big\}<\beta<\max_j\Big\{-\frac{L_k(\zeta^1\vert\theta_j)}{L_k(\zeta^2\vert\theta_j)}\Big\}.
\end{align}
% \end{small}%
Finally, after selecting an appropriate $\beta$, we utilize \eqref{e1}, \eqref{e2} to obtain the values of $p_{k,1},p_{k,2}$ that determine the fake likelihood functions with 
$x_1$ satisfying $x^+>x_1>x_1'$ if $d_k>0$ and $x^-<x_1<x_1'$ if $d_k<0$ and $x_2$ is given by \eqref{x2p}. 
\qedsymb
\section{Proof of Corollary \ref{multiagent_cor}}
\label{a24}
We observe from Lemma \ref{Lem_uninformative} that there is no choice of $\widehat{L}_k(\cdot\vert\theta_1),\widehat{L}_k(\cdot\vert\theta_2)$ that simultaneously satisfies \eqref{lem_ratiosm30} for $j=1$ and $j=2$ for an adversary $k$ with non-informative PMFs. Thus, if we set $\widehat{L}_k(\theta_1)=\widehat{L}_k(\theta_2)=L_k(\theta_1)=L_k(\theta_2)$ for all adversaries $k\in\mathcal{N}^m$ with uninformative PMFs, then for these agents we have $R_{k,j}=0$ for all $j\in\{1,2\}$. 

From Theorem \ref{region}, we have that \eqref{lem_ratiosm30} is satisfied for $j=1,2$ for every adversary $k\in\mathcal{N}^m$ with informative PMFs if it follows the construction given by Theorem \ref{region} for $\epsilon$ satisfying \eqref{e_condition}. Let $\mathcal{N}^{m,+}$ denote the set of adversaries with informative PMFs. Then, for the above choices of fake likelihood functions, by utilizing \eqref{lem_ratiosm30} and summing up over all adversaries $k\in\mathcal{N}^{m,+}$ we have
\begin{align}
\label{in1}
    &\sum_{k\in\mathcal{N}^{m,+}}R_{k,1}>|\mathcal{N}^{m,+}|S_1\overset{(a)}\geq S_1\\
    \label{in2}
    &\sum_{k\in\mathcal{N}^{m,+}}R_{k,2}>|\mathcal{N}^{m,+}|S_2\overset{(a)}\geq S_2
\end{align}
for $\epsilon$ given by \eqref{e_multiagent}. In the above inequalities $(a)$ is true due to the fact that $S_1,S_2\geq0$ as they are nonnegative weighted sums of KL divergences. We can include in the summations in the LHS of the inequalities \eqref{in1}, \eqref{in2} the adversaries with uninformative PMFs as well, because for the choice $\widehat{L}_k(\theta_1)=\widehat{L}_k(\theta_2)=L_k(\theta_1)=L_k(\theta_2)$, we have $R_{k,j}=0$ for $j=1,2$. As a result, \eqref{in1}, \eqref{in2} hold as well if we change the summation to be over $\mathcal{N}^m$ instead of $\mathcal{N}^{m,+}$. This implies that  \eqref{inequality_re}, or equivalently \eqref{lem_ratiosm}, is satisfied for $j=1,2$, which implies that the network is mislead for both $\theta^{\star}=\theta_1$ and $\theta^{\star}=\theta_2$.\qedsymb
\section{Proof of Lemma \ref{one_variable}}
\label{a25}
In this result we provide specific conditions under which the construction \eqref{construction10} with one free variable (i.e., $p_{k,1}=p_{k,2}=p$) is insufficient to mislead the network for both states, thus showing that such a construction cannot guarantee that the network will be deceived for both states in the general case. 

Following the proof of Theorem \ref{region}, the region $\mathcal{R}_{k}$ is characterized by the lines \eqref{line1}, \eqref{line2} with intersection point at $(x_1',x_2')$ with $x_1'<0,x_2'>0$ if $d_{k}<0$ and $(x_1',x_2')$ with $x_1'>0,x_2'<0$ if $d_{k}>0$ (for more details {\em see} proof of Theorem \ref{region}). By restricting the choice to one free variable (i.e., $p_{k,1}=p_{k,2}=p$), we restrict the values of $x_1,x_2$ to be given by the line
\begin{align}
    x_2=-x_1
\end{align}
for $x_1,x_2$ satisfying \eqref{x_1}, \eqref{x_2}. 

We also observe that the intersection points of lines $r_1(x_1),r_2(x_1)$ (defined in \eqref{line1}, \eqref{line2}), which define $\mathcal{R}_k$, with the $x_2$ axis satisfy $r_1(0)>0,r_2(0)<0$. Thus, $x_2=-x_1$ does not intersect $\mathcal{R}_{k}$ in the following cases:
\begin{align}
    &-\frac{L_{k}(\zeta^1\vert\theta_1)}{L_{k}(\zeta^2\vert\theta_1)}<-1,\quad \text{if} \quad d_{k}<0\\
    &-\frac{L_{k}(\zeta^1\vert\theta_2)}{L_{k}(\zeta^2\vert\theta_2)}<-1, \quad\text{if} \quad d_{k}>0.
\end{align}
This implies that in this case there are no choices of $x_1,x_2\in\mathcal{R}_k$ such that $x_2=-x_1$, which means that 
there are no choices of $p_{k,1},p_{k,2}$ such that $p_{k,1}=p_{k,2}=p$ mislead the network for both $\theta^{\star}=\theta_1$ and $\theta^{\star}=\theta_2$.\qedsymb
\section{Proof of Lemma \ref{uniform_rem}}
\label{a3}
Since adversary $k\in\mathcal{N}^m$ has informative PMFs, then there exists at least a $\zeta'$ such that:
\begin{align}
\label{eqd1}
    L_{k}(\zeta'\vert\theta_1)>L_{k}(\zeta'\vert\theta_2)
\end{align}
or
\begin{align}
\label{eqd2}
    L_{k}(\zeta'\vert\theta_1)<L_{k}(\zeta'\vert\theta_2).
\end{align}
We focus on the first case (i.e., \eqref{eqd1}) and the second case (i.e., \eqref{eqd2}) follows accordingly. Since \eqref{eqd1} holds, $\zeta'\in\mathcal{D}^1_k$. Let us assume that there exists no $\zeta\in\mathcal{Z}_{k}$ such that $\zeta\in\mathcal{D}^2_k$, meaning that there exists no $\zeta\in\mathcal{Z}_{k}$ such that $L_{k}(\zeta\vert\theta_1)<L_{k}(\zeta\vert\theta_2)$. Then,
\begin{align}
    \sum_{\zeta}L_{k}(\zeta\vert\theta_1)=1>\sum_{\zeta}L_{k}(\zeta\vert\theta_2)
\end{align}
which cannot hold, since $L(\cdot\vert\theta)$, $\theta\in\Theta$ is a PMF. Thus, there exists at least one observation $\zeta\in\mathcal{Z}_k$ such that $L_{k}(\zeta\vert\theta_1)<L_{k}(\zeta\vert\theta_2)$, which means that there exists at least one observation $\zeta\in\mathcal{Z}_k$ that belongs to $\mathcal{D}^2_{k}$. Following the same reasoning for the other case (\eqref{eqd2} holds), we arrive at the same conclusion. Thus, both sets $\mathcal{D}^1_k$ and $\mathcal{D}^2_{k}$ are always non-empty for an agent with informative PMFs.\qedsymb
\section{Proof of Theorem \ref{opt_attackth}}
\label{a32}
\noindent 
The objective function \eqref{opt_altao} yields
% \begin{small}
\begin{align}
    \label{exp_costao_pr}
    &\frac{1}{2}\sum_{\theta\in\Theta}\mathcal{C}(\theta)=\frac{1}{2}\Big[\sum_{k\in\mathcal{N}^n}u_kD_{KL}\Big(L_k(\theta_1))||L_k(\theta_2)\Big)\nonumber\\
    &+\sum_{k\in\mathcal{N}^m}u_{k}\sum_{\zeta} L_{k}(\zeta|\theta_1)\log\frac{\widehat{L}_k(\zeta|\theta_1)}{\widehat{L}_k(\zeta|\theta_2)}
    \nonumber\\
    &+\sum_{k\in\mathcal{N}^n}u_kD_{KL}\Big(L_k(\theta_2)||L_k(\theta_1)\Big)\nonumber\\
    &+\sum_{k\in\mathcal{N}^m}u_{k}\sum_{\zeta} L_{k}(\zeta|\theta_2)\log\frac{\widehat{L}_k(\zeta|\theta_2)}{\widehat{L}_k(\zeta|\theta_1)}\Big]\nonumber\\
    &=\frac{1}{2}\Bigg(\sum_{k\in\mathcal{N}^n}u_kD_{KL}\Big(L_k(\theta_1)||L_k(\theta_2)\Big)\nonumber\\
    &+\sum_{k\in\mathcal{N}^n}u_kD_{KL}\Big(L_k(\theta_2)||L_k(\theta_1)\Big)\nonumber\\
    &+\sum_{k\in\mathcal{N}^m}u_{k}\sum_{\zeta}\Big(( L_{k}(\zeta|\theta_1)-L_{k}(\zeta|\theta_2))\log\widehat{L}_k(\zeta|\theta_1)\nonumber\\
    &-(L_{k}(\zeta|\theta_1)-L_{k}(\zeta|\theta_2))\log\widehat{L}_k(\zeta|\theta_2)\Big)\Bigg).
\end{align}
% \end{small}\noindent
First, let us define
\begin{align}
Z_{k}(\zeta)\triangleq L_{k}(\zeta|\theta_1)-L_{k}(\zeta|\theta_2),\quad\zeta\in\mathcal{Z}_{k}.
\end{align}
We observe that minimizing 
\eqref{exp_costao_pr} reduces to the following minimization problem since these are the only terms depending on $\widehat{L}_{k}(\cdot\vert\theta_1)$ and $\widehat{L}_{k}(\cdot\vert\theta_2)$:
{\begin{small}{
\begin{align}
\label{exp_costb}
    &\min_{\widehat{L}_{k}(\cdot\theta_1),\widehat{L}_{k}(\cdot\vert\theta_2)}\sum_{k\in\mathcal{N}^m}u_{k}\sum_{\zeta\in\mathcal{Z}_k}\Big[Z_{k}(\zeta)\log\widehat{L}_k(\zeta|\theta_1)\nonumber\\
    &-Z_{k}(\zeta)\log\widehat{L}_k(\zeta|\theta_2)\Big]\\
    %\label{inequality1}
    &\text{s.t.} \quad\epsilon\leq\widehat{L}_{k}(\zeta|\theta_1),\quad\quad\quad\quad\forall\zeta\in\mathcal{Z}_{k},\,\forall k\in\mathcal{N}^m\\
    %\label{inequality2}
    &\quad\quad\,\epsilon\leq\widehat{L}_{k}(\zeta|\theta_2),\quad\quad\quad\quad\forall\zeta\in\mathcal{Z}_{k},\,\forall k\in\mathcal{N}^m\\
    %\label{equality_con1}
    &\quad\quad\,\sum_{\zeta\in\mathcal{Z}_{k}}\widehat{L}_{k}(\zeta|\theta_1)=1,\,\,\,\quad\forall k\in\mathcal{N}^m\\
    %\label{equality_con2}
    &\quad\quad\,\sum_{\zeta\in\mathcal{Z}_{k}}\widehat{L}_{k}(\zeta|\theta_2)=1,\quad\,\,\,\forall k\in\mathcal{N}^m.
\end{align}
}\end{small}}\noindent
The minimization problem is {\em separable} across agents $k\in\mathcal{N}^m$. 
Hence, we can compute the optimal $\widehat{L}_{k}(\cdot)$ for each malicious agent $k\in\mathcal{N}^m$ independently. Thus, we have
{\begin{small}{
\begin{align}
\label{min_pr}
    &\min_{\widehat{L}_{k}(\cdot\vert\theta_1),\widehat{L}_{k}(\cdot\vert\theta_2)}\sum_{\zeta\in\mathcal{Z}_k}\Big[Z_{k}(\zeta)\log\widehat{L}_k(\zeta|\theta_1)-Z_{k}(\zeta)\log\widehat{L}_k(\zeta|\theta_2)\Big]
    \\
    \label{inequality1}
    &\text{s.t.}\quad \epsilon\leq\widehat{L}_{k}(\zeta|\theta_1),\quad\forall \zeta\in\mathcal{Z}_{k}\\
    \label{inequality2}
    &\quad\quad\,\epsilon\leq\widehat{L}_{k}(\zeta|\theta_2),\quad\forall \zeta\in\mathcal{Z}_{k}\\
    \label{equality_con1}
    &\quad\quad\,\sum_{\zeta\in\mathcal{Z}_{k}}\widehat{L}_{k}(\zeta|\theta_1)=1\\
    \label{equality_con2}
    &\quad\quad\,\sum_{\zeta\in\mathcal{Z}_{k}}\widehat{L}_{k}(\zeta|\theta_2)=1.
\end{align}
}\end{small}}\noindent
We further observe from \eqref{min_pr}-\eqref{equality_con2} that the minimization is separable with respect to $\widehat{L}_k(\cdot\vert\theta_1)$ and $\widehat{L}_k(\cdot\vert\theta_2)$. 
As a result, we can decompose the optimization problem \eqref{min_pr} as follows:
\begin{align}
\label{opt_prao2}
    &\min_{\widehat{L}_{k}(\theta_1)}\sum_{\zeta\in\mathcal{Z}_k}Z_{k}(\zeta)\log\widehat{L}_{k}(\zeta|\theta_1)\nonumber\\
    &-\max_{\widehat{L}_{k}(\theta_2)}\sum_{\zeta\in\mathcal{Z}_k}Z_{k}(\zeta)\log\widehat{L}_{k}(\zeta|\theta_2)
\end{align}
% \end{small}\noindent
subject to constraints \eqref{inequality1}-\eqref{equality_con2}.

We focus on the minimization problem and the same rationale applies to the maximization problem.  
Let us define
% \begin{small}
\begin{align}
&J_{k}(\widehat{L}_{k}(\cdot\vert\theta_1))=\sum_{\zeta\in\mathcal{Z}_{k}}Z_{k}(\zeta)\log\widehat{L}_{k}(\zeta|\theta_1)\\
&J^+_{k}(\widehat{L}_{k}(\cdot\vert\theta_1))=\sum_{\zeta\in\mathcal{D}^1_k}Z_{k}(\zeta)\log\widehat{L}_{k}(\zeta|\theta_1)\\
&J^-_{k}(\widehat{L}_{k}(\cdot\vert\theta_1))=\sum_{\zeta\in\mathcal{D}^2_k}Z_{k}(\zeta)\log\widehat{L}_{k}(\zeta|\theta_1)
\end{align}
% \end{small}\noindent
where 
$\mathcal{D}^1_{k}=\{\zeta\in\mathcal{Z}_k:Z_{k}(\zeta)\geq0\}$, and $\mathcal{D}^2_{k}=\mathcal{Z}_{k}\setminus\mathcal{D}^1_{k}$. 
Then, the minimization problem in \eqref{opt_prao2} can be rewritten as
% \begin{small}
\begin{align}
\label{min_praltb}
    &\min_{\widehat{L}_{k}(\cdot\vert\theta_1)} J^+_{k}(\widehat{L}_{k}(\cdot\vert\theta_1))+J^-_{k}(\widehat{L}_{k}(\cdot\vert\theta_1))
    \\
    \label{inequality1b}
    &\text{s.t.}\quad\epsilon\leq\widehat{L}_{k}(\zeta|\theta_1),\quad\forall \zeta\in\mathcal{Z}_{k}\\
    \label{equality_con1b}
    &\quad\quad\,\sum_{\zeta\in\mathcal{Z}_{k}}\widehat{L}_{k}(\zeta|\theta_1)=1.
\end{align}
% \end{small}\noindent
Let us introduce $\sigma\triangleq\sum_{\zeta\in\mathcal{D}^1_k}\widehat{L}_{k}(\zeta|\theta_1)$. The above minimization problem can be equivalently rewritten as
\begin{align}
\label{min_pralte}
    &\min_{\widehat{L}_{k}(\cdot\vert\theta_1),\sigma} J^+_{k}(\widehat{L}_{k}(\cdot\vert\theta_1),\sigma)+ J^-_{k}(\widehat{L}_{k}(\cdot\vert\theta_1),\sigma)
    \\
    \label{inequality1c}
    &\text{s.t.}\quad\epsilon\leq\widehat{L}_{k}(\zeta|\theta_1),\quad\forall \zeta\in\mathcal{Z}_{k}\\
    \label{equality_con1c}
    &\quad\quad\,\sum_{\zeta\in\mathcal{D}^1_{k}}\widehat{L}_{k}(\zeta|\theta_1)=\sigma\\
    &\quad\quad\,\sum_{\zeta\in\mathcal{D}^2_{k}}\widehat{L}_{k}(\zeta|\theta_1)=1-\sigma.
\end{align}
We will first show by contradiction that the optimal value for $\sigma$ is $\sigma^{\star}=|\mathcal{D}^1_{k}|\epsilon$ (note that this is the minimum admissible value for $\sigma$ due to constraints \eqref{inequality1c}). First, note that $\sigma^{\star}=|\mathcal{D}^1_{k}|\epsilon\iff\widehat{L}^{\star}_{k}(\zeta\vert\theta_1)=\epsilon$ for all $\zeta\in\mathcal{D}^1_{k}$.

Let us suppose that $\sigma^{\star}$ is not optimal. This means that there exists a $\sigma'>\sigma^{\star}$ such that $\sigma'$ is optimal and an optimal $\widehat{L}_{k}'(\cdot)$ such that $J_{k}(\widehat{L}_{k}'(\cdot),\sigma')$ 
attains its minimum value. Then,
% \begin{small}
\begin{align}
    \sigma'>\sigma^{\star}\implies \widehat{L}_{k}'(\tilde{\zeta}\vert\theta_1)>
    \epsilon, \text{for some } \tilde{\zeta}\in\mathcal{D}^1_{k}.
\end{align}
% \end{small}\noindent
Let us denote $\widehat{L}_{k}'(\tilde{\zeta}\vert\theta_1)=\epsilon+\tilde{\epsilon}$ for some $\tilde{\epsilon}>0$. In the following, we will construct a counterexample to show that there always exist  $\widehat{L}_{k}''(\cdot),\sigma''$ 
such that $J_{k}(\widehat{L}_{k}''(\cdot),\sigma')<J_{k}(\widehat{L}_{k}'(\cdot),\sigma')$ with $\sigma''<\sigma'$. Thus, $\sigma'$ cannot be optimal and as a result, the optimal value of $\sigma$ should be $\sigma^{\star}=|\mathcal{D}^1_{k}|\epsilon$. We construct $\widehat{L}_{k}''(\cdot\vert\theta_1)$ as follows:
% \begin{small}
\begin{align}
\label{construction}
    \widehat{L}_{k}''(\zeta\vert\theta_1)=\begin{cases}\widehat{L}_{k}'(\tilde{\zeta}\vert\theta_1)-\tilde{\epsilon},&\text{if } \zeta=\tilde{\zeta},\\ 
    \widehat{L}_{k}'(\widehat{\zeta}\vert\theta_1)+\tilde{\epsilon},&\text{for some } \widehat{\zeta}\in\mathcal{D}^2_{k},\\
    \widehat{L}_{k}'(\zeta\vert\theta_1), &\forall \zeta\in\mathcal{Z}_k\setminus\{\tilde{\zeta},\widehat{\zeta}\}.
    \end{cases}
\end{align}
% \end{small}\noindent
Note that the above construction satisfies the constraints of the optimization problem. Also, note that $\sigma''<\sigma'$. Then, since we assumed that $J_{k}(\widehat{L}_{k}'(\cdot\vert\theta_1),\sigma')$ is optimal, we have
% \begin{small}
\begin{align}
\label{contradiction1}
    &J_{k}(\widehat{L}_{k}'(\cdot\vert\theta_1),\sigma')\leq J_{k}(\widehat{L}''_{k}(\cdot\vert\theta_1),\sigma'')\nonumber\\
    &\iff
    J^+_{k}(\widehat{L}_{k}'(\cdot\vert\theta_1),\sigma')+J^-_{k}(\widehat{L}_{k}'(\cdot\vert\theta_1),\sigma')
    \nonumber\\
    &\leq J^+_{k}(\widehat{L}''_{k}(\cdot\vert\theta_1),\sigma'')+J^-_{k}(\widehat{L}''_{k}(\cdot\vert\theta_1),\sigma'')\nonumber\\
    &\overset{(a)}\iff
    Z_{k}(\tilde{\zeta})\log\widehat{L}_{k}'(\tilde{\zeta}\vert\theta_1)+Z_{k}(\widehat{\zeta})\log\widehat{L}_k'(\widehat{\zeta}\vert\theta_1)\nonumber\\
    &\leq Z_{k}(\tilde{\zeta})\log(\widehat{L}_{k}'(\tilde{\zeta}\vert\theta_1)-\tilde{\epsilon})+Z_{k}(\widehat{\zeta})\log(\widehat{L}_{k}'(\widehat{\zeta}\vert\theta_1)+\tilde{\epsilon})
\end{align}
% \end{small}\noindent
where $(a)$ is true because $Z_{k}(\zeta)\log\widehat{L}_k'(\zeta\vert\theta_1)=Z_{k}(\zeta)\log\widehat{L}_{k}''(\zeta\vert\theta_1)$ for all $\zeta\neq\tilde{\zeta},\widehat{\zeta}$ by construction of $\widehat{L}_{k}''(\cdot\vert\theta_1)$ ({\em see} \eqref{construction}). We observe that \eqref{contradiction1} is a contradiction, because $\log(\cdot)$ is monotonically increasing, $Z_{k}(\tilde{\zeta})$ is non-negative and $Z_{k}(\widehat{\zeta})$ is negative. Thus, the optimal value of $\sigma$ is $\sigma^{\star}=|\mathcal{D}^1_{k}|\epsilon$. This implies that the optimal value of $\widehat{L}_{k}(\zeta\vert\theta_1)$ for every $\zeta\in\mathcal{D}^1_{k}$ is
\begin{align}
\label{solution_min1}
    \widehat{L}^{\star}_k(\zeta|\theta_1)=\epsilon,\quad\forall \zeta\in\mathcal{D}^1_{k}.
\end{align}
Then, we proceed to find $\widehat{L}^{\star}(\zeta|\theta_1)$, for every $\zeta\in\mathcal{D}^2_{k}$.
% \begin{small}
\begin{align}
\label{min_praltd}
    &\min_{\widehat{L}_{k}(\cdot\vert\theta_1)}J^-_{k}(\widehat{L}_{k}(\cdot\vert\theta_1),\sigma^{\star})
    \\
    \label{inequality1d}
    &\text{s.t.}\quad\epsilon\leq\widehat{L}_{k}(\zeta\vert\theta_1),\quad\forall \zeta\in\mathcal{D}^2_k\\
    \label{equality_con1d}
    &\quad\quad\,\sum_{\zeta\in\mathcal{D}^2_k}\widehat{L}_{k}(\zeta|\theta_1)=1-|\mathcal{D}^1_{k}|\epsilon.
\end{align}
$J^-_{k}(\widehat{L}_{k}(\cdot\vert\theta_1),\sigma^{\star})$ is a convex function (as it is a non-negative sum of convex functions, due to the fact that $Z_k(\zeta)<0$ for all $\zeta\in\mathcal{D}^2_k$) and since the constraints are affine, Slater's condition is satisfied and hence, the KKT conditions are necessary and sufficient conditions for the solution of the optimization problem \cite{Boyd}. We introduce the Langrange multipliers $\lambda_{\zeta}$, for every $\zeta\in\mathcal{D}^2_k$ for the inequality constraints \eqref{inequality1d} and $v$ for the equality constraint \eqref{equality_con1d}, respectively. The KKT conditions are given by
\begin{align}
\label{eq1}
    &\frac{\partial\mathcal{L}}{\partial\widehat{L}_{k}(\zeta|\theta_1)}=\frac{L_{k}(\zeta|\theta_1)-L_{k}(\zeta|\theta_2)}{\widehat{L}_{k}(\zeta|\theta_1)}-\lambda_{\zeta}+v=0\\
    \label{eq3}
    &\lambda_{\zeta}(\widehat{L}_{k}(\zeta|\theta_1)-\epsilon)=0\\
    \label{eq4}
    &\sum_{\zeta}\widehat{L}_{k}(\zeta|\theta_1)=1-|\mathcal{D}^1_{k}|\epsilon\\
    &\lambda_{\zeta}\geq0
\end{align}
Let $\lambda_{\zeta}=0$ for all $\zeta\in\mathcal{D}^2_{k}$. Then,  
\eqref{eq1} yields
\begin{align}
    \widehat{L}_{k}(\zeta|\theta_1)=-\frac{L_{k}(\zeta|\theta_1)-L_{k}(\zeta|\theta_2)}{v}
\end{align}
Substituting to \eqref{eq4} we get
\begin{align}
    v=-\frac{\sum_{\zeta\in\mathcal{D}^2_{k}}L_{k}(\zeta|\theta_1)-L_{k}(\zeta|\theta_2)}{1-|\mathcal{D}^1_{k}|\epsilon}
\end{align}
Thus,
% \begin{small}
\begin{align}
\label{solution_min}
    \widehat{L}^{\star}_{k}(\zeta|\theta_1)=\frac{(L_{k}(\zeta|\theta_1)-L_{k}(\zeta|\theta_2))(1-|\mathcal{D}^1_{k}|\epsilon)}{\sum_{\zeta\in\mathcal{D}^2_{k}}L_{k}(\zeta|\theta_1)-L_{k}(\zeta|\theta_2)},\,\forall\zeta\in\mathcal{D}^2_{k}.
\end{align}
This concludes the proof for the optimal optimal values of $\widehat{L}_{k}(\zeta|\theta_1)$, $\zeta\in\mathcal{Z}_k$. Working in the same way for the maximization problem in \eqref{opt_prao2}, we obtain the optimal values for $\widehat{L}_{k}(\zeta|\theta_2)$, $\zeta\in\mathcal{Z}_k$, as well.\qedsymb
\section{Proof of Theorem \ref{efficiciency_Th3}}
\label{a33}
We introduce the following variables:
\begin{align}
\label{barx1}
    &\bar{x}_1\triangleq\log\frac{\epsilon}{1-|\mathcal{D}^1_{k}|\epsilon}\\
    \label{barx2}
    &\bar{x}_2\triangleq\log\frac{1-|\mathcal{D}^2_{k}|\epsilon}{\epsilon}.
\end{align}
By replacing $\log\frac{\epsilon}{1-|\mathcal{D}^1_{k}|\epsilon}$ and $\log\frac{1-|\mathcal{D}^2_{k}|\epsilon}{\epsilon}$ with $\bar{x}_1$ and $\bar{x}_2$, respectively, in \eqref{app_c1} and \eqref{app_c2}, we get the following system of inequalities:
\begin{align}
\label{cs11}
    &\bar{x}_2\hspace{-1mm}>\hspace{-1mm}\frac{S_1-u_{k}(c_{k,1}+b_{k,1})}{u_{k}\xi_{k,1}}-\frac{\sigma_{k,1}}{\xi_{k,1}}\bar{x}_1\\
    \label{cs22}
    &\bar{x}_2\hspace{-1mm}<\hspace{-1mm}-\frac{S_2+u_{k}(c_{k,2}+b_{k,2})}{u_{k}\xi_{k,2}}-\frac{\sigma_{k,2}}{\xi_{k,2}}\bar{x}_1
\end{align}
Now, if we allow $\bar{x}_1,\bar{x}_2$ to take arbitrary values in $\mathbb{R}$ instead of satisfying \eqref{barx1}, \eqref{barx2}, then there always exist $\bar{x}_1,\bar{x}_2$ that satisfy \eqref{cs11}, \eqref{cs22} for $\frac{\sigma_{k,1}}{\xi_{k,1}}\neq\frac{\sigma_{k,2}}{\xi_{k,2}}$ (which is always true for an adversary with informative PMFs). Let this space of values of $\bar{x}_1,\bar{x}_2\in\mathbb{R}$ that satisfy \eqref{cs11} and \eqref{cs22} be denoted as $\bar{\mathcal{R}}_k$. The lines that define this region $\bar{\mathcal{R}}_k$ are given by
\begin{align}
\label{x21line}
    &\bar{x}_2=\frac{S_1-u_{k}(c_{k,1}+b_{k,1})}{u_{k}\xi_{k,1}}-\frac{\sigma_{k,1}}{\xi_{k,1}}\bar{x}_1\\
    \label{x22line}
    &\bar{x}_2=-\frac{S_2+u_{k}(c_{k,2}+b_{k,2})}{u_{k}\xi_{k,2}}-\frac{\sigma_{k,2}}{\xi_{k,2}}\bar{x}_1
\end{align}
where $\bar{x}_1,\bar{x}_2\in\mathbb{R}$.

However, $\bar{x}_1,\bar{x}_2$ cannot take arbitrary values and we see from \eqref{barx1}, \eqref{barx2}, that their relation is given by
\begin{align}
\label{relation_x1x2}
    \bar{x}_2=\log(e^{-\bar{x}_1}+|\mathcal{D}^1_{k}|-|\mathcal{D}^2_{k}|).
\end{align}
Note that the argument of $\log(e^{-\bar{x}_1}+|\mathcal{D}^1_{k}|-|\mathcal{D}^2_{k}|)$ is always positive for $0<\epsilon<\frac{1}{|\mathcal{D}^j|}$, $j=1,2$, which is ensured due to \eqref{condition_e}.

Next, we answer the question of whether there are sufficiently small values of $\epsilon$ such that $\bar{x}_1,\bar{x}_2$ given by \eqref{barx1}, \eqref{barx2} satisfy \eqref{cs11} and \eqref{cs22}, or equivalently $\bar{x}_1,\bar{x}_2\in\bar{\mathcal{R}}_k$. If the answer is positive, then this implies that for sufficiently small values of $\epsilon$, the strategy presented in Theorem \ref{opt_attackth} misleads the network for both $\theta^{\star}=\theta_1$ and $\theta^{\star}=\theta_2$.

Under the strategy presented in Theorem \ref{opt_attackth} the values of $\bar{x}_1,\bar{x}_2$ are represented by \eqref{relation_x1x2}, which approaches the line $\bar{x}_2=-\bar{x}_1$, as $\bar{x}_1$ gets smaller. 
We observe that if the slopes of the lines \eqref{x21line}, \eqref{x22line} satisfy
\begin{align}
    &-\frac{\sigma_{k,1}}{\xi_{k,1}}>-1\iff\sigma_{k,1}<\xi_{k,1}\\
    &-\frac{\sigma_{k,2}}{\xi_{k,2}}<-1\iff\sigma_{k,2}>\xi_{k,2}
\end{align}
then, the region $\bar{\mathcal{R}}_k$ always contains $\bar{x}_1,\bar{x}_2$ that satisfy $\bar{x}_2=-\bar{x}_1$. Since \eqref{relation_x1x2} approaches the line $\bar{x}_2=-\bar{x}_1$ as $\bar{x}_1$ gets smaller, there always exist $\bar{x}_1,\bar{x}_2$ that are contained in $\bar{\mathcal{R}}_k$ for sufficiently small $\bar{x}_1$. Finally, $\bar{x}_1$ is monotonically increasing in $\epsilon$, which implies that there is always a sufficiently small $\epsilon^{\star}$ such that $\bar{x}_1$ and $\bar{x}_2$ given by \eqref{barx1} and \eqref{barx2}, respectively, satisfy \eqref{cs11} and \eqref{cs22} for every $\epsilon<\epsilon^{\star}$.\qedsymb
% you can choose not to have a title for an appendix
% if you want by leaving the argument blank
% \section{}
% Appendix two text goes here.

% use section* for acknowledgment
% \section*{Acknowledgment}

% The authors would like to thank...

% Can use something like this to put references on a page
% by themselves when using endfloat and the captionsoff option.
\ifCLASSOPTIONcaptionsoff
  \newpage
\fi

\end{document}